%% file: NoteVer1.tex
\documentclass[journal]{IEEEtran}
\usepackage{microtype}%if unwanted, comment out or use option "draft"
\usepackage[ruled]{algorithm2e} % For algorithms
\usepackage{tikz}
\usepackage{xcolor}
\usepackage{pgfplots}
\usepackage{amsmath}
\usepackage{amsthm}
\usepackage{amssymb}
\usepackage{verbatim}
\usepackage{graphicx}
\usepackage{ifthen}
\usepackage{color}
\definecolor{bblue}{HTML}{4F81BD}
\definecolor{rred}{HTML}{C0504D}
\definecolor{ggreen}{HTML}{9BBB59}
\definecolor{ppurple}{HTML}{9F4C7C}

\input{input.tex}

\newboolean{showcomments}
\setboolean{showcomments}{true}
\newcommand{\jk}[1]{  \ifthenelse{\boolean{showcomments}}
{ \textcolor{red}{(JK says:  #1)}} {}  }
\newcommand{\rv}[1]{  \ifthenelse{\boolean{showcomments}}
{ \textcolor{red}{(RV says:  #1)}} {}  }

\newcommand{\ra}{\rightarrow}

\newcommand{\ignore}[1]{}
\begin{document}
\title{Speed Scaling On Parallel Servers with MapReduce Type Precedence Constraints} 
 
 \author{Rahul Vaze and~Jayakrishnan Nair% <-this % stops a space
\thanks{Rahul Vaze is with the School of Technology and Computer Science, Tata Institute of Fundamental Research, India. Jayakrishnan Nair is with the Department of Electrical Engineering, IIT Bombay, India.}}% <-this % stops a space

\maketitle

\begin{abstract}
  A multiple server setting is considered, where each server has
  tunable speed, and increasing the speed incurs an energy cost. Jobs
  arrive to a single queue, and each job has two types of sub-tasks,
  map and reduce, and a {\bf precedence} constraint among them: any
  reduce task of a job can only be processed once all the map tasks of
  the job have been completed. In addition to the scheduling problem,
  i.e., which task to execute on which server, with tunable speed, an
  additional decision variable is the choice of speed for each server,
  so as to minimize a linear combination of the sum of the flow times
  of jobs/tasks and the total energy cost. The precedence constraints
  present new challenges for the speed scaling problem with multiple
  servers, namely that the number of tasks that can be executed at any
  time may be small but the total number of outstanding tasks might be
  quite large. We present simple speed scaling algorithms that are
  shown to have competitive ratios, that depend on the power cost
  function, and/or the ratio of the size of the largest task and the
  shortest reduce task, but not on the number of jobs, or the number
  of servers.
\end{abstract}

\input{Introduction.tex}

\input{SystemModel.tex}

\input{JobByJobFlowTime.tex}
\input{OfflineUnsplittableMapJobs.tex}
\input{WmaxbywminGuarantee}
\input{NewProof.tex}

\input{Simulations}
\section{Conclusions}
In this paper, we have made progress in deriving  online algorithms for scheduling and speed scaling jobs with MapReduce precedence constraints. Without the precedence constraints there is a large body of literature on efficiently solving the flow time plus energy problem, however, very little is known with precedence constraints, that is inherently a hard problem. The algorithms we proposed followed a simple SRPT policy, (process as many shortest jobs/tasks as possible)  
and the total power used by all the servers is equal to the total number of outstanding tasks, which has a intuitive appeal in the sense that it balances the energy and the delay cost effectively. The derived competitive ratio guarantee depends on the power function $P$, 
and/or the ratio of the size of the largest map/reduce task and the smallest reduce task, that is typically small for MapReduce applications. 
%It is worthwhile noting that even when this ratio is $1$, no competitive ratio guarantee was not available in prior work. 

\bibliographystyle{IEEEtran}
\bibliography{../../refs, ../../extrarefs, ../../refsnew}
\input{AppJobFlow}
\input{applemjump1}
\input{Lemmajumpphi3}
\input{appfinalcasetaskbased}
\input{AppProofTaskBasedSS.tex}
\input{lemlpnorm}
\input{AppMSalpha.tex}
\end{document}

%% file: input.tex
\newtheorem{theorem}{Theorem}
\newtheorem{assumption}[theorem]{Assumption}
\newtheorem{definition}[theorem]{Definition}
\newtheorem{remark}[theorem]{Remark}
\newtheorem{proposition}{Proposition}
\newtheorem{lemma}{Lemma}

\def\bb0{{\mathbb{0}}}

\def\bb{{\mathbf{b}}}

\def\b0{{\mathbf{0}}}
\def\opt{\mathsf{OPT}}

\def\b1{{\mathbf{1}}}

\def\bbE{{\mathbb{E}}}

\def\bbN{{\mathbb{N}}}

\def\bbR{{\mathbb{R}}}

\def\bbZ{{\mathbb{Z}}}

\def\cJ{\mathcal{J}}

\def\sfD{\mathsf{D}}

\def\sfJ{\mathsf{J}}

\def\sfc{{\mathsf{c}}}

\def\sfff{{\mathsf{f}}}

\def\sfm{{\mathsf{m}}}
\def\sfn{{\mathsf{n}}}

\def\sfr{{\mathsf{r}}}

\def\sf0{{\mathsf{0}}}

\def\nn{\nonumber}

%% file: Introduction.tex
\section{Introduction}

In distributed/parallel processing systems such as MapReduce
\cite{dean2008mapreduce} or Hadoop \cite{shvachko2010hadoop}, Dryad
\cite{isard2007dryad}, jobs/flows alternate between computation and
communication stages, where a new stage cannot start until all the
required tasks/flows have been processed in the preceding
stage. Essentially, there are precedence constraints between different
tasks, e.g. until all the map tasks of a job are not completed, no
reduce task of that job can be started.

A typical figure of merit in these systems is the delay seen by a
job/task, where a job consists of multiple map and reduce tasks, and
the precedence constraints present a new set of challenges for
deriving optimal routing and scheduling policies that minimize the sum
of the job/task delays. There are typically three different metrics
that are considered with multiple jobs, {\it makespan}, i.e., the
finish time of the last job, {\it completion time}, i.e., the sum of
finish times (counted from time $0$) of all jobs, and {\it flow time},
i.e., the sum of response times (finish minus the arrival time) of all
jobs. Solving the makespan problem is typically easier than the
completion time problem and solving the flow time problem is the
hardest.

The MapReduce scheduling problem has been considered in prior work
quite extensively for both single and multiple server systems,
however, under the typical assumption that the server speeds are
fixed. In particular, for the offline setting, where non-causal job
arrival information is assumed, \cite{moseley2011scheduling,
  chang2011scheduling, tan2012performance, zheng2013new,
  zhu2014minimizing, yuan2014joint, wang2016maptask} considered either
the makespan or the completion time minimization problem, and derived
approximation algorithms mostly assuming that any map task can be
split arbitrarily to allow parallel execution on multiple servers. The
more challenging online setting has been considered in
\cite{moseley2011scheduling, chang2011scheduling, zheng2013new,
  luo2017online}, where very recently \cite{chen2017online} found an
online algorithm with competitive ratio of $3$ for the makespan
minimization problem assuming as before that any map task can be split
and executed in parallel on multiple servers. There is some work on
energy efficient MapReduce scheduling
\cite{yigitbasi2011energy,mashayekhy2014energy} with heuristic
algorithms.

Makespan and flow time minimization has also been considered for general precedence constraints represented by an acyclic directed
graphs (DAG)  with {\bf fixed speed} multiple servers   \cite{ graham1966bounds, blazewicz1983scheduling,chudak1999approximation,
  pruhs2008speed}. Most recent work in this model does not even assume
knowing the size of the jobs when they arrive \cite{ngarg} (called the
non-clairvoyant model).
%Moreover, simpler paradigms are mostly considered, such as any map task can be split in as any many pieces to allow parallel execution on multiple servers \cite{zhu2014minimizing}, or simpler metrics  such as makespan (offline \cite{zhu2014minimizing}, online \cite{chen2017online}), and not the flow time, which is defined as the sum of the response times (completion minus arrival time) of each job. Recall that in MapReduce framework, no reduce task is ever allowed to be split. 

In this paper, we consider a parallel server setting, where
each server has a tuneable speed. Corresponding to the speed of
operation, there is a power/energy cost for each server identified by
function $P$ which is typically assumed to have the form
$P(s) = s^\alpha$ for $\alpha >1$.  Clearly, increasing the speed of
the server reduces the flow time but incurs a larger energy
cost. Thus, there is a natural tradeoff between the the {\it flow
  time} and the total energy cost, and a natural objective with
tuneable servers is to minimize a linear combination of the flow time
and total energy, called \emph{flow time plus energy}.

Jobs arrive over time where each job has a fixed number of map and
reduce tasks, and any reduce task of a job can be executed only after all the
map tasks belonging to the job are completed. Both preemption and
migration are allowed, i.e., a task can be preempted on one server and
restarted on another server later, which is a natural requirement with
MapReduce constraints, since map tasks have inherently higher
priority. Moreover, no task can be split, and thus cannot be processed
on multiple servers at the same time.  

To keep the problem most
general, we assume an arbitrary input setting, where both the arrival
times of jobs, and sizes (of tasks) are arbitrary (can be chosen by an
adversary), and the objective is to find optimal online algorithms
(that use only causal information) in terms of the competitive
ratio. Competitive ratio is defined to be ratio of the cost incurred
by the online algorithm to the cost of the optimal offline algorithm
$\opt$ that knows the entire input in advance, maximized over all
possible inputs.

%Thus, in this paper, we consider the {\it online} problem of 
%scheduling, and speed scaling in a multi-server setting to minimize
%the flow time plus energy, where jobs arrive (are released) over time
%and decisions have to be made causally, under the {\bf MapReduce} type precedence constraints. 
%On the arrival of a new job, a
%centralized controller needs to make a causal decision about which
%jobs to process on which server and at what speed, %The
%model, however, does not allow job splitting, i.e., a job can only be
%processed on a single server at any time.

In prior work, there is a large body of work on online algorithms for
the flow time plus energy problem, however, to the best of our
knowledge except \cite{pruhs2008speed, bampis2014energy}, does not
consider any precedence constraints.  Without precedence constraints,
online algorithms for multiple servers to minimize flow time and
energy with constant competitive ratios have been derived in
\cite{lam2008competitive, greiner2009bell,
  bansal2009speedconf,Andrew2010, gupta2010scalably, lam2012improved,
  gupta2012scheduling, SpeedScalingOptFairRobust,devanur2018primal},
under both the homogenous server model, i.e., $P(.)$ is identical for
all servers, and the heterogenous model.

The most relevant work for our considered problem is
\cite{pruhs2008speed, megow2013dual, bampis2014energy}, that considers multiple
servers with tuneable speed, and the objective is to find an optimal
algorithm that minimizes the completion time under a total energy
constraint for arbitrary precedence constraints defined by a directed
acyclic graph (DAG) \cite{pruhs2008speed} and for MapReduce type
constraints \cite{megow2013dual, bampis2014energy}. However, the results of
\cite{pruhs2008speed} are limited for the case when all jobs are
available at time $0$, and its approximation ratio scales as square of
the logarithm of the number of servers. In \cite{megow2013dual, bampis2014energy}, 
an offline problem (where 
the exact arrival sequence of jobs is known a priori)  to minimize the completion time with 
 MapReduce precedence constraints is considered. Moreover, \cite{bampis2014energy}
requires that for each task there is a preassigned server that only is allowed to process it. A
power function dependent approximation ratio has been derived in \cite{bampis2014energy} with
energy augmentation, improving upon \cite{megow2013dual} that
considered only a single server setting.

In this paper, our endeavour is to only consider the MapReduce type of
precedence constraints and find online algorithms for the flow time
plus energy problem with competitive ratios that do not scale with
either the number of jobs or the number of servers.
%lam2012improved, gupta2012scheduling, SpeedScalingOptFairRobust,devanur2018primal}.
The main difficulty with multiple servers under precedence constraints
is that at certain time, there might be a large number of outstanding
tasks (few map but large number of reduce tasks), but very small (less than the number of
servers) number of tasks that can be executed (e.g. only map tasks),
making some servers idle. Essentially, all the non-triviality stems in
controlling this event, since we are looking for a sample path result
against the optimal offline algorithm ($\opt$) that might keep all
servers busy making the comparison between any proposed online
algorithm and $\opt$ difficult.

With precedence constraints, flow time can be counted in two ways,
job-wise (departure time of the last reduce task - arrival time of
job) or task-wise. We consider both these settings in this paper. For
the job-wise setting, our results are limited for the case when all
jobs arrive at time $0$ similar to \cite{pruhs2008speed}. For the
task-wise flow time, we consider the general online setting, where
jobs arrive over time and the algorithm has only causal information.
 
Throughout this paper, we consider the shortest remaining processing
time (SRPT) algorithm for executing the outstanding executable tasks
on multiple servers.\footnote{The choice of SRPT is motivated from its
  optimality in single server environments (see
  \cite{shrage1968proof}), and its near-optimal performance in
  multi-server environments (\cite{leonardi2007approximating,vaze, grosof2018srpt}).}
In particular, let $n$ be the total number of outstanding tasks, and
$k\le K$ be the number of tasks that can be executed at time $t$. Then
with $K>1$ servers, the $\min\{k, K\}$ tasks with the shortest
remaining size are executed on the $\min\{k, K\}$ servers.
 %Throughout, we also make an assumption that the size of any map task is at most as much as the size of any reduce task, which is justified for the setup that map tasks are less demanding than reduce tasks, which have to combine work across multiple map tasks. It is worth noting that even with this assumption the problem remains challenging, since handling the precedence constraints is the major bottleneck. 
%Thus, at any time, the algorithm processes as many map tasks using SRPT subject to maximum $K$.  

With $K$ homogenous servers with power function $P(s) = s^\alpha$, our
contributions are as follows.\footnote{Generalizable all convex $P$
  functions satisfying Assumption \ref{ass:Pconvex}.}
\begin{enumerate}
\item For the job-wise flow time + energy problem, the proposed
  algorithm achieves a competitive ratio of at most $4(2-1/K)^\alpha +o(1)$.

\item For the task-wise flow time + energy problem, we first propose
  an algorithm that achieves the competitive ratio of at most
  $$ P(2-1/K) (2\beta +2+2(\alpha-1)) +o(1),$$ where
  $$\beta = \frac{w_{\max}^t}{w_{reduce, \min}},$$ and $w_{\max}^t$ and
  $w_{reduce, \min}$ is the maximum size of any map/reduce task and
  the minimum size of any reduce task, respectively.  Ideally, the competitive ratio should have no dependence on the size of the map or reduce tasks, but the
  competitive ratio of the proposed algorithm is a linear function of $\beta$ (the ratio of the
  size of the largest map/reduce task and the shortest reduce task).
  For practical MapReduce applications, $\beta$ typically small (since
  map tasks corresponding to any job are designed to have nearly equal
  size, as are reduce tasks \cite{hammoud2011locality}), and hence
  the derived guarantee is still meaningful.  Moreover, from a theoretical
  point of view, as far as we know, no competitive ratio result is
  known in prior work even when $\beta=1$, i.e. all tasks have equal
  size.

\item For the task-wise flow time + energy problem, for the power function $P(s) = s^\alpha$, we remove the dependence of $\beta$ on the competitive ratio next, where we show that the proposed algorithm achieves a competitive ratio of at most $8+\frac{4}{2-\alpha}+3^{\alpha}$, however, the result holds for only $1 < \alpha < 2$. 

  \end{enumerate}

Since we adapt the SRPT algorithm to work with MapReduce constraints in this paper, we next review and contrast this work 
with our recent prior work \cite{vaze} that considers the online SRPT algorithm without precedence constraints. 
In \cite{vaze},  the competitive ratio of the 
SRPT algorithm with multiple servers without any precedence constraints has been shown to to be upper bounded by $P(2-1/m) \left(2 + \frac{2}{P^{-1}(1)}
    \max(1,P(\bar{s}))\right),$ where $\bar{s}$ is a
    constant associated with the function $P(\cdot).$
     To compare the results derived in this paper to that of \cite{vaze}, we note that for the job based flow time + energy (offline) problem, adapting SRPT algorithm to work with 
     MapReduce constraints only changes the competitive ratio by an additive $o(1)$ term. Thus, SRPT is adaptable with MapReduce constraints and incurs a very small penalty in this case.
     
     For the task based flow time + energy problem, compared to \cite{vaze}, using SRPT algorithm in the presence of MapReduce constraints there is an additional penalty of $\beta = \frac{w_{\max}^t}{w_{reduce, \min}}$, and $w_{\max}^t$ in the competitive ratio,
  $w_{reduce, \min}$ is the maximum size of any map/reduce task and
  the minimum size of any reduce task. Typically, $\beta =\frac{w_{\max}^t}{w_{reduce, \min}}$ is small \cite{hammoud2011locality} for MapReduce applications, the derived competitive ratio guarantee is still meaningful and adapting SRPT algorithm with MapReduce algorithm is efficient. 
  
  In terms of technical novelty compared to \cite{vaze}, enforcing
  MapReduce constraints bring in new challenges while using the SRPT
  algorithm. In particular, for the job based flow time + energy
  problem, we introduce a new job-SRPT algorithm that is clearly
  sub-optimal but helps in analytical tractability. In particular, we
  prove Lemma \ref{lem:optjobsrpt}, which allows us to prove the same
  competitive ratio bound (upto $o(1)$ term) as in \cite{vaze} even
  when there are arbitrary precedence constraints between different
  tasks of any one job and not just MapReduce constraints.
  
  For the task based flow time + energy problem, as long as there are at least  $K$ (number of servers) map tasks, one can directly use results from \cite{vaze}. 
  The main technical difficulty arises when there are less than $K$ map tasks but a large number of reduce tasks (which cannot be processed because of the precedence constraint). This presents a unique challenge in multi-server systems, that there are a large number of outstanding tasks but some of the servers are idling. Since we are considering worst case input, we have to analyze every single sample path.
   To deal with this case, we  construct a novel potential function \eqref{defn:new:phi3}, and derive its drift in Lemma \ref{lem:drfitspecial}, 
   and show that the competitive ratio of the SRPT algorithm in presence of the MapReduce constraints for the  task based flow time + energy problem, increases by a factor of $\beta =\frac{w_{\max}^t}{w_{reduce, \min}}$ compared to the SRPT algorithm without precedence constraints  \cite{vaze}. 
   
   To remove the dependence of the competitive ratio on $\beta$, we next consider a small tweak to the SRPT algorithm that allows multiple tasks to be 
  processed simultaneously by a single server, and  show that we can get a competitive ratio independent of the instance of the problem $\beta$ that only depends on the value of $\alpha$ where $P(s)=s^\alpha$ as long as $\alpha<2$. The main technical contribution is the construction of potential function \eqref{defn:phiss} and the analysis of its derivative ($d\Phi_4(t)/dt$) in Appendix \ref{app:onlinems} similar to Lemma \ref{lem:drfitspecial}.

%% file: SystemModel.tex
\section{System Model}
We consider a MapReduce profile of job arrivals, where each job $i$
has $m_i$ map tasks and $r_i$ reduce tasks, where any of the reduce
tasks of a job can be executed only after the completion of all
the~$m_i$ map tasks.  Let the input consist of a set of jobs $\cJ$
with cardinality $\sfJ,$ where job $i\in \cJ$ arrives (is released) at
time~$a_i,$ with the work/size of its $k^{th}$ map task being
$\sfm_{ik}, k =1, \dots, m_i,$ and that of its $\ell^{th}$ reduce task
being $\sfr_{i\ell}, \ell =1, \dots, r_i$. The set of tasks belonging
to job~$i$ is denoted by $\text{Job}_i.$ The total number of map
(respectively, reduce) tasks summed across all jobs is denoted by
$\sfn_m$ (respectively, $\sfn_r$).

Note that we are not assuming that the size of the map and reduce tasks are identical. 
The model is general, where a map or reduce task can have an arbitrary size, and derive competitive ratio guarantees that are either dependent or independent of the map/reduce task sizes.

There are $K$ homogenous parallel servers, each with the same power
function $P(s),$ where $P(s)$ denotes the power consumed while running
at speed $s$.
% Typically, $2\le \alpha\le 3$.
Any task can be processed by any of the~$K$ servers, and both
preemption and task migration are allowed, which is essential with
precedence constraints (since map tasks have inherently higher
priority, and reduce tasks need to be migrated).

%The new feature that is useful for many modern applications such as MapReduce etc. is to include precedence constraints among sub-jobs execution. In particular, when a job arrives, it consists of multiple sub-jobs that have precedence constraints of the type NotBefore, that is a sub-job $j$ cannot be executed before finishing job $i$.  This presents a new challenge in scheduling and speed scaling since simple SRPT algorithm cannot be executed.
\begin{definition}\label{defn:free}
  A task is defined to be {\it free} at time $t$, if it can be
  processed at time $t$. Thus, a map task is always free, while a
  reduce task is free at time $t$ if all the map tasks belonging to the same job
  have been completed at time $t$. A reduce task that is not free is called {\it
    caged}.
\end{definition}

%We define a job (among the reduce
%To overcome this challenge, we define a sub-job $j$ to be {\it free} if there is no other outstanding sub-job $i$ for which job $j$ cannot be executed before sub-job $i$. A job that is not free is called {\it caged}. 
Let $c_{ik}$ be the time at which the task (map/reduce) $k$ of job $i$
is completed. Then, the flow time $\sfff_{ik}$ for a map/reduce task
$k$ of job $i$ is defined as $\sfff_{ik} = c_{ik}-a_i$ (completion
time minus the arrival time).  We consider the two possible scenarios
of counting flow times that correspond to counting the delay seen by a
job, or the sum of the delays seen by all tasks within a job.  For the
job-by-job case, we define the flow time of job $i$ as
$\sfff_i = T_{i}-a_{i}$, where
$T_{i} = \max_{k \in \text{Job}_i} c_{ik}$ is the time at which the
last reduce task of job $i$ finishes, and total flow time is defined
as $F_{\text{job}}= \sum_{i \in \cJ} \sfff_{i}$. For the task-by-task
case, $\sfff_{i} = \sum_{k\in \text{Job}_i} \sfff_{ik}$ and the total
flow time is $F_{\text{task}}= \sum_{i \in \cJ} \sfff_{i}$.

$F_{\text{job}}$ corresponds to counting per-job delay cost, while the
$F_{\text{task}}$ reflects per-bit delay cost or average flow
time. With both these definitions, the following expressions for the
flow time are useful.
$$F_{\text{job}}  = \int n(t) dt, \quad F_{\text{task}} = \int (n_f(t) + n_c(t)) dt,$$ where at time $t$,
$n(t)$ is the number of outstanding jobs, while $n_f(t)$ and $n_c(t)$
is the number of outstanding free and caged tasks, respectively.
%From here on we refer to
%$F$ as just the flow time.  

%For the $\opt$, we add the superscript $o$. Note that the total flow time can also be written as $F = \int (n_f(t) + n_c(t)) dt$, which will be useful in the sequel. 
%We appropriately define flow time if the delay is counted for each job and not its sub-jobs separately in Section \ref{sec:jobdelay}.

Let server $k$ run at speed $s_k(t)$ at time $t$.  The energy cost is
defined as $\sum_{k=1}^KP(s_k(t))$ integrated over the flow time.
Choosing larger speeds reduces the flow time, however, increases the
energy cost, and the natural objective function that has been
considered extensively in the literature is the sum of flow time and
energy cost, which we define as

\begin{align}\label{eq:costjob}
  C_{\text{job}} &=  \int n(t) dt + \int  \sum_{k=1}^K P(s_k(t))dt, \\ \label{eq:costtask}
  C_{\text{task}} &=  \int (n_f(t) + n_c(t)) dt + \int  \sum_{k=1}^K P(s_k(t))dt.
\end{align} 

Note that we have added the two costs without weighting them, since
the weight can be absorbed in the power function $P$.  Any online
algorithm only has causal information, i.e., it becomes aware of job
$i$ only at time $a_i$.  Using only this causal information, any
online algorithm has to decide at what speed each server should be run
at at each time.  Let the job arrival sequence be $\sigma = \{(a_i), (\sfm_{ik}), (\sfr_{i\ell}), \ i\in \cJ\}$.
For the online algorithm $\sigma$ is revealed causally, while for an offline algorithm $\sigma$ is assumed to be known in advance. 
Let the cost \eqref{eq:costjob} of an online
algorithm $A$ be $C^A$, and the cost for the offline optimal algorithm
$\opt$ be $C^\opt$. 
Then the worst case competitive ratio of the online
algorithm $A$ is defined as
$\sfc_A = \max_{\sigma}\frac{C^A(\sigma)}{C^\opt(\sigma)}$, 
and the objective function considered in this paper is to find an
online algorithm that minimizes the worst case competitive ratio
$ \sfc^\star = \min_A \sfc_A$.
%Similar, competitive ratio definition holds if we consider the task based cost $C_{\text{task}}$.

A typical approach in speed scaling literature to upper bound the
competitive ratio $\sfc_A$ for an algorithm A is via the construction
of a potential function $\Phi(t)$ and show that for any input
sequence~$\sigma,$ 
\begin{equation}\label{eq:mothereq} 
n(t) + \sum_{k =1}^K P(s_k(t)) +  \frac{d\Phi(t)}{dt}  \le \sfc_A\biggl(n^o(t)
  + \sum_{k\in \opt} P({\tilde s}_k(t))\biggr),
\end{equation}
whenever $\frac{d\Phi(t)}{dt}$ exists, and ${\tilde{s}_k}$ is the
speed of server $k$ with the $\opt$ and $n^o$ stands for the number of
outstanding jobs with the $\opt$, and that $\Phi(t)$ satisfies the
following {\it boundary} conditions,
\begin{enumerate}
\item Before any job arrives and after all jobs are finished,
  $\Phi(t)= 0$, and
\item $\Phi(t)$ does not have a positive jump discontinuity at any
  point of non-differentiability.
  % increase on any job arrival or departure by the algorithm or the
  % $\opt$, i.e., there are no positive jumps at discontinuities in
  % $\Phi(t)$.
\end{enumerate} 
For the task based flow time, $n(t)$ is replaced by $n_f(t) +n_c(t)$
in \eqref{eq:mothereq}.  Then, integrating \eqref{eq:mothereq} with
respect to $t$, we get that
%\vspace{-0.15in}
\begin{align*}\label{eq:mothereq1} & \int \left(  n(t) + \sum_{k =1}^K P(s_k(t))\right)  \le \sfc_A \int \biggl(n^o(t)
  + \sum_{k\in \opt} P({\tilde s}_k(t))\biggr),
\end{align*} 
which is equivalent to showing that $C^{A}(\sigma) \le c_A \
C^{\opt}(\sigma)$ for any input $\sigma$ as required. 

%In this paper, we consider three pertinent input models, 
%and for each of them,%We begin with the easiest  case in Section \ref{sec:split}, also the most popular in literature, where each map tasks is allowed to be split and run
%parallely on multiple servers, while no reduce task can be split. 
\begin{remark}\label{rem:boundary} Suppose that the second boundary
  condition is not satisfied, and let $\Phi(t)$ increase by amount
  $D_j$ at the $j^{th}$ discontinuous point, such that the total
  increase is $\sum_{j} D_j \le \sfD\ C^{\opt}$.  If
  \eqref{eq:mothereq} holds at points where $\Phi(t)$ is
  differentiable, then we get that
  $C^{A}(\sigma) \le (\sfc_A +\sfD )\ C^{\opt}(\sigma)$ for any input
  $\sigma$, making the upper bound on the competitive ratio $\sfc_A+\sfD$.
\end{remark}
%We begin with offline case in next Section, when all jobs arrive at time $0$, which is generalized in Section \ref{sec:onlineunsplit}, where jobs arrive at arbitrary times.  The basic technical idea for all the following sections is to propose a speed scaling algorithm and construct a potential function for which \eqref{eq:mothereq} holds for an appropriate choice of $c$. 
%%Next, in Section \ref{sec:nosplitzero}, we consider the more challenging case, when neither any map or reduce task can be split, but to keep the exposition 
%%simple, 
%In Section \ref{sec:onlineunsplit}, for the proposed algorithm, the boundary conditions will not hold as stated above, but the following modified boundary conditions  will hold. \begin{enumerate}
%\item Before any job arrives and after all jobs are finished,
%  $\Phi(t)= 0$, and
%\item 
%  % increase on any job arrival or departure by the algorithm or the
%  % $\opt$, i.e., there are no positive jumps at discontinuities in
%  % $\Phi(t)$.
%\end{enumerate} 
%Then, if \eqref{eq:mothereq} holds where $d\Phi(t)/dt$ exists, integrating \eqref{eq:mothereq} with respect to $t$, we get that
%
%\begin{equation}\label{eq:mothereqjump}
%C_A \le  \int c\biggl(n_o(t) + \sum_{k =1}^K P(s^o_k(t))\biggr)dt + \sfD C_{\opt},
%\end{equation} 
One easy lower bound on the $\opt$'s cost can be obtained by removing the precedence constraints, and considering that the arrival sequence is such that
no job/task has to wait behind any other job/task. Thus, the cost of $\opt$ for any job/task with size $w$ is $\min_s w/s + (w/s) P(s)$. The optimal speed $s$ satisfies $1+P(s^\star) = s^\star P'(s^\star)$ and the optimal per job/task cost is $w P'(s^\star)$. 
With $P(s)=s^2$, $s^\star=1$, and the minimum flow time + energy for any job/task is $2 w$. Counting across all jobs/tasks, we get the following lower bound.
\begin{proposition}\label{prop:optlb} For the $\opt$, 
$C_{\text{job}}^\opt(\sigma) \ge P'(s^\star) (\sum_{j\in \cJ} w_j)$ and 
$C_{\text{task}}^\opt(\sigma) \ge P'(s^\star) (\sum_{j \in \cJ} (\sfm_{jk} + \sfr_{j\ell}))$, where $1+P(s^\star) = s^\star P'(s^\star)$.
% = 2 (\sum_{i=1}^{\sfn_m} \sfm_i + \sum_{j=1}^{\sfn_r} \sfr_j)$, where $\sfm_i$ and $\sfr_j$ is the size of the $i^{th}$ map and $j^{th}$ reduce task. 
\end{proposition}

Throughout this paper, we make the following assumptions about the
power function $P(.)$.
\begin{assumption}
  \label{ass:Pconvex}
  $P:\bbR_+ \ra \bbR_+$ with $P(0) = 0$ is a differentiable, strictly
  increasing, and strictly convex function, which implies $\lim_{s \ra
    \infty}P(s) = \infty,$ and $\bar{s}:=\inf\{s > 0 \ |\ P(s) > s \}
  < \infty$. Moreover, for $x,y > 0,$ $P(xy) \leq P(x) P(y).$ We also
  assume that $\Delta(x) = P'(P^{-1}(x)) = o(x)$ as $x \ra \infty.$
\end{assumption}

Assumption \ref{ass:Pconvex} is satisfied by power functions of the
form $P(s) = c s^{\alpha},$ where $\alpha > 1$ and $c \geq 1.$ For
simplicity, we will use $P(s) = s^\alpha$ (where $\alpha > 1$) in the
rest of the paper (our argument generalize easily to power functions
that satisfy Assumption~\ref{ass:Pconvex}); in this case, $\Delta(x) =
\alpha x^{1-1/\alpha}.$ This assumption is needed for application of Lemma \ref{lem:bansal} that is useful in bounding the derivative of the potential functions.
\begin{assumption}
  \label{ass:size}
  In Section~\ref{sec:job}, where we consider the job-based flow time
  metric, we assume a natural \emph{many jobs} scaling, where the total number
  of jobs $\sfJ$ is large, while the maximum size of any job
  $w_{\max}$ and the number of servers $K$ satisfy $K w_{\max}
  \Delta\left(\sfJ/K\right) = o (\sum_{i=1}^\sfJ w_j)$, which is quite a mild requirement and easily satisfied in practice. For $P(s) =
  s^\alpha,$ this means $K^{1/\alpha} w_{\max} \sfJ^{1-1/\alpha} =
  o(\sum_{i=1}^\sfJ w_j)$.
  In Section~\ref{sec:task}, where we consider task-based flow time
  metric, we assume a natural \emph{many tasks} scaling where the total number
  of map plus reduce tasks $\sfn_m +\sfn_r$ is large, and $O(\sfJ
  w^t_{\max} K^2 r_{\max}^{1-\frac{1}{\alpha}}) =
  o(\sum_{j \in \cJ} (\sfm_{jk} + \sfr_{j\ell}))$
  %(\sum_{i=1}^{\sfn_m} \sfm_i + \sum_{j=1}^{\sfn_r} \sfr_j)$,
  where, $w^t_{\max}$ denotes the size of the
  largest map/reduce task, and $r_{\max}$ is the maximum number of
  reduce tasks any job can have. 
  %\jk{Abuse of notation --- $m_i$ meant
    %the number of map tasks for job~$i$ before, but here it means the
    %size of a certain map task.}
\end{assumption}
We essentially need Assumption \ref{ass:size} to couple the increase in the 
potential functions at the points of discontinuities, and the total cost of the $\opt$.

%% file: JobByJobFlowTime.tex
\section{Job Based Flow Time}\label{sec:job}
In this section, we consider the job-by-job flow time and the cost
metric is $C_{\text{job}}$ \eqref{eq:costjob}. We restrict ourselves
to the case that all jobs are available at time $0$.  The objective is
to propose an algorithm with a constant approximation (competitive)
ratio. This is typically a difficult task since there are intra-job
precedence constraints, and the flow time is counted as the departure
time of the last reduce task of a job minus its arrival time, coupling
the processing of different tasks of the system.

To overcome this difficulty, we propose a job-by-job SRPT algorithm,
called job-SRPT, as follows.  For any task (map or
reduce) $k$ of job $j$, let $w_{kj}(t)$ be its total remaining work at
time $t$.  Then the total remaining {\bf cumulative work} of job $j$ at time $t$ is
defined as $w_j(t) = \sum_{k \in \text{Job}_j} w_{kj}(t)$ (sum of the
sizes of all its map and reduce tasks). Index the $n(t)$ outstanding
jobs at time $t$ in increasing order of their remaining cumulative work. If
$n(t)\ge K$, the job-SRPT algorithm processes the $K$ shortest jobs on
the $K$ servers, and on each server the shortest free task of each job
is executed. Thus, each server is executing a task corresponding to a
different job at each time. Otherwise, if $n(t)< K$, then the $n(t)$
distinct jobs (shortest free task of that job) are executed on the
$n(t)$ servers, and for the rest of $K-n(t)$ servers, the $K-n(t)$ 
shortest free tasks (other than the ones already executing) across all
outstanding $n(t)$ jobs are executed.  For the job-SRPT algorithm, we
propose the following speed. For server $k$,
\begin{equation}\label{eq:speeddef}
s_k(t) =
\begin{cases}
 &   P^{-1}\left(\frac{n(t)}{K}\right)\  \text{if}\  n(t) \ge K, \\
&  P^{-1}(1), \ \ \ \ \ \ \ \  \ \text{otherwise}.
\end{cases}
\end{equation}
%In Fig. \ref{fig:jobSRPT}, we illustrate the  job-SRPT algorithm. 
The main result of this section for the job-based flow time + energy problem is as follows.
\begin{theorem}\label{thm:srptimproved} 
  Under Assumptions~\ref{ass:Pconvex} and \ref{ass:size},
  the job-SRPT algorithm with speed scaling \eqref{eq:speeddef} for job based flow time problem
  has competitive ratio $ P(2-1/K) \left(2 + \frac{2}{P^{-1}(1)} \max(1,P(\bar{s}))\right) + o(1)$, where $\bar{s}$ is defined in Assumption \ref{ass:Pconvex}.
\end{theorem}
Taking $P(s) = s^{\alpha}$ for $\alpha >1,$ the competitive ratio
equals $4(2-1/K)^{\alpha} + o(1).$ It is worth noting that this
competitive ratio is same as with the no precedence constraints case
in \cite{vaze} with multi-server SRPT algorithm up to a $o(1)$ term.

Clearly, the job-SRPT algorithm is not optimal since one can easily
construct an example where executing multiple tasks from the same job
on different servers is better, however, the job-SRPT algorithm allows analytical
tractability because of the following lemma.

\begin{lemma}\label{lem:optjobsrpt} 
  With MapReduce precedence constraints, when all jobs are available
  at time $0$, for $K$ fixed speed servers, $\opt$ that follows the 
  job-SRPT algorithm with server speed $(2-1/K)$ has a sum of job flow times
  that is at most that under $\opt$ with $K$ servers having speed $1$
  for any job input sequence $\sigma$.
\end{lemma} 
This lemma is similar to the claim for classical SRPT
\cite{phillips1997optimal}, where each job is a single autonomous
entity without any precedence constraints, and follows by proving
Lemma \ref{lem:scalingsrpt} and Lemma \ref{lem:monotone}.
Importantly, the claim in \cite{phillips1997optimal} is true for
classical SRPT even in the online case, when jobs are released over
time. However, with precedence constraints, it is not true for the
job-SRPT algorithm (Lemma \ref{lem:monotone} no longer holds) in the
online case. The main reason is that, if at any time there is a single
job with multiple tasks that are being processed simultaneously, an
arrival of new job does not guarantee that the number of jobs departed
by any future time is at least as many as the ones would have departed
if the new job would not have arrived.

How Lemma \ref{lem:optjobsrpt} helps is as follows. In light of Lemma
\ref{lem:optjobsrpt}, we can let $\opt$ execute the job-SRPT algorithm
with an arbitrary unknown speed and claim that it is
$P(2-1/K)$-approximate with respect to $\opt$ (executes arbitrary
scheduling and speed choice). Note that $\opt$ will process all tasks
of a job at the same speed because of convexity of $P$. Thus, for a
job $j$, let the speed $\opt$ uses be $s_j$. Then for $\opt$ following
job-SRPT, the speed of job $j$ (whenever any task of job $j$ is
scheduled) is chosen to be $(2-1/K)s_j$. Thus, $\opt$ following
job-SRPT will have at most the same flow-time as the $\opt$ from Lemma
\ref{lem:optjobsrpt}, with an additional energy cost of
$P(2-1/K)$. Thus, $\opt$ following job-SRPT is $P(2-1/K)$-approximate
with respect to $\opt$.
  
For the rest of the section, we call $\opt$ following job-SRPT as $\opt$. 
The main advantage of considering job-SRPT for both the algorithm and the $\opt$, is that we can get rid of all the precedence constraints within a job, since at any time only the $K$ distinct shortest (cumulative) jobs are being processed if the number of outstanding jobs is at least $K$. Since all jobs are available at time $0$, 
the case that multiple tasks belonging to the same job are processed simultaneously occurs only when the number the number of outstanding jobs is less than $K$, which can be handled separately.

Using job-SRPT algorithm for scheduling, we use the same potential function that has been proposed in \cite{vaze} to analyze multi-server SRPT algorithm. 
Recall that the remaining cumulative size of any job be the sum of the remaining sizes of its map and reduce tasks.
At time $t$, let $n^o(t,q)$ and $n(t,q)$ denote the number of
unfinished jobs with the $\opt$ and the algorithm, respectively, with
remaining cumulative size at least $q$. In particular, $n^o(t,0) = n^o(t)$ and
$n(t,0) = n(t)$. Let
$$d(t,q) = \max\left\{0, n(t,q) - n^o(t,q)\right\}.$$ 

Let $f:\{i/K:\ i \in \bbZ_+\} \ra \bbR_+$ is
  defined as follows: $f(0) = 0,$ and for $i \geq 1,$ $f\left(\frac{i}{K}\right)
  = f\left(\frac{i-1}{K}\right) + \Delta\left(\frac{i}{K}\right)$, where $\Delta(x) :=
  P'(P^{-1}(x)).$
We consider the potential function
\begin{align}\label{defn:phi}
\Phi(t) &= \Phi_1(t)+\Phi_2(t), \\
 \Phi_1(t) &= c_1\int_{0}^\infty f\left(\frac{d(t,q)}{K}\right) dq, \\
\Phi_2(t) &= c_2 \int_{0}^\infty (n(t,q) - n^o(t,q)) dq,
\end{align}
where $c_1, c_2$ are positive constants whose value will be specified later. 
$\Phi_1$ part is sufficient for the case when the $n(t)\ge K$, while $\Phi_2$ is needed to handle the case when $n(t) < K$.

%To upper bound  the drift $d\Phi_1/dt$, we would like to leverage the results from \cite{vaze} (Lemma \ref{lemma:jobphi1_OPT-SRPT}) derived for multi-server SRPT algorithm, however, there is a technical problem described as follows.   When this happens for the online algorithm, we can just count the negative contribution in $d\Phi_1/dt$ from only one of the task of each job and disregard other tasks'  contribution belonging to the same job. This is not permissible for the $\opt$, however. Next,. 

\begin{remark}\label{rem:jumpoff} 
As long as there are $K$ distinct jobs with both the algorithm and the $\opt$, because of processing by the algorithm, $d(t+dt,q) = d(t,q)-1$ for $q \in [q_k - s_k(t) dt, q_k]$, where $q_k, k=1,\dots,K$ is the remaining (cumulative) size of the shortest $K$ jobs. Similarly, for the $\opt$, $d(t,q) = d(t,q)+1$ for $q \in [q^o_k - {\tilde s}_k(t) dt, q^o_k]$,  for $q^o_k, k=1,\dots,K$. Thus, one can use the results from  \cite{vaze} (Lemma \ref{lemma:jobphi1_OPT-SRPT}) directly when  both the algorithm and the $\opt$ have $K$ distinct jobs, to bound the drift $d\Phi_1/dt$, that are valid for the case when each job has a single task, and no job can be split and processed parallely on multiple servers.
When either the algorithm or the $\opt$ has less than $K$ jobs available, the job-SRPT algorithm, may process more than one task belonging to the same job simultaneously on multiple servers and \cite{vaze} (Lemma \ref{lemma:jobphi1_OPT-SRPT}) is not directly applicable.  
Since the algorithm's contribution to the drift $d\Phi_1/dt$ is negative, we can just count one term $d(q) = d(q)-1$ for $q \in [q_k - s_k(t) dt, q_k]$ for each of the distinct jobs and disregard the 
contribution in drift $d\Phi_1/dt$ from multiple tasks of a same job that are being processed simultaneously. 
This, however, cannot be done for the $\opt$, since $\opt$'s contribution is positive in drift $d\Phi_1/dt$, and %Recall that $\opt$ follows job-SRPT, hence it processes as many distinct jobs (upto $K$) at any time. 
%Thus, as long as there are $K$ distinct jobs with the $\opt$, all the $K$ servers are processing distinct jobs, in which case, Lemma \ref{lemma:jobphi1_OPT-SRPT}  for the drift of $\Phi_1$ holds from 
%Lemma 3 \cite{vaze} since this case is similar to the no precedence case. A technical problem can arise, when the number of jobs with the $\opt$ is less than $K$, in which case multiple tasks from the same job may be processed by the job-SRPT algorithm on two or more servers. In this case, Lemma 3 \cite{vaze} does not hold, since it is similar to job splitting for which Lemma 3 \cite{vaze} is not applicable. 
to handle this case, we introduce deletion of tasks for the $\opt$ (defined next) which will make $\Phi_1$ and $\Phi_2$ increase discontinuously, and we will count the total increase in $\Phi_1(t)$ and $\Phi_2$ at all points of discontinuities. 
\end{remark}
{\bf Deleting tasks with the $\opt$:} Let the time instant at which $\opt$ has $k$ outstanding jobs be $t_k$, where $k=K-1,K-2, \dots, 1$. Note that $t_k$'s need not be distinct.  Let at time $t_k$, the set of jobs $j$ for which the $\opt$ is about to process more than one task parallely be $S_k$. Then at $t_k^-$, for all jobs $j\in S_k$, we delete all its reduce tasks and all the map tasks other than its shortest map task. This will ensure that there is at most one task of each job being processed by any of the servers throughout,  allowing the applicability of  Lemma 3 \cite{vaze}. The deletion of tasks will result in a total upward jump in $\Phi_1(t)$ and $\Phi_2(t)$ which can be bounded as follows (proof in Appendix \ref{app:lemjump1}).

\begin{lemma}\label{lem:jump1} The total upward jump in
  the potential function $\Phi_1(t)$ and $\Phi_2(t)$ because of deletion of tasks with the $\opt$ is $\sum_{j=1}^{K-1} c_1 w_j \Delta(\frac{\sfJ}{K}) + \sum_{j=1}^{K-1}c_2 w_j$, where $w_j$ are the sizes of potentially the $K-1$ largest (cumulative) size jobs, and  $\sfJ$ is the total number of jobs in the input.
\end{lemma}

Next, we state Lemma 3  and Lemma 4 from \cite{vaze} that  
 allows us to bound the derivative of $\Phi_1(t)$ and $\Phi_2(t)$ as long as $\opt$ does not process more than one task belonging to the same job on multiple servers, and we count contribution of only distinct jobs of the algorithm in $d\Phi_1(t)/dt$ and $d\Phi_2(t)/dt$.
\begin{lemma}
  \label{lemma:jobphi1_OPT-SRPT} 
When $\opt$ does not process more than one task belonging to the same job simultaneously on multiple servers, for $n \geq K$, $d\Phi_1/dt \le  c_1 n^o - c_1 n + c_1 \left(\frac{K-1}{2}\right) + c_1 \sum_{k\in \opt} P(\tilde{s}_k)$,
  while for $n < K,$
    $d\Phi_1/dt \le c_1 n^o - c_1 \frac{n(n+1)}{2K} + c_1 \sum_{k\in \opt} P(\tilde{s}_k)$.
\end{lemma}
\begin{lemma}
  \label{lemma:jobphi2_OPT-SRPT}
  $d\Phi_2/dt \leq -c_2 \min(K,n) P^{-1}(1)  +  c_2\sum_{k\in \opt} \max\{P(\bar{s}), P({\tilde s}_k)\}.$

\end{lemma}
Now, we prove Theorem \ref{thm:srptimproved}. 
From Proposition \ref{prop:optlb}, note that the minimum cost $C_{\text{job}}$ \eqref{eq:costjob} of $\opt$ is at least $P(s^\star) \sum_{j=1}^\sfJ w_j$. From Lemma \ref{lem:jump1}, we have that the total jump size at all points of discontinuities is $\sum_{j=1}^{K-1} c_1 w_j \Delta(\frac{\sfJ}{K}) + \sum_{j=1}^{K-1}c_2 w_j$.  From  Assumption \ref{ass:Pconvex}, $\Delta(\frac{\sfJ}{K}) = o(\sfJ/K)$, e.g. for $P(s)=s^\alpha$, $\Delta(x) = O(x^{1-1/\alpha})$.
Thus, $D = o(1)$ from Assumption \ref{ass:size} and Remark \ref{rem:boundary}. Thus, if we can show that \eqref{eq:mothereq} holds for some $c$, we will have that the total competitive ratio is upper bounded by $c+o(1)$.
From Remark \ref{rem:jumpoff}, we ensure that the $\opt$ never processes more than one task belonging to the same job, thus applying, Lemma \ref{lemma:jobphi1_OPT-SRPT} and \ref{lemma:jobphi2_OPT-SRPT}, (similar to \cite{vaze}) we get 
  \eqref{eq:mothereq} 
 \begin{align*}
  n + \sum_{k \in A} P(s_k) + d\Phi(t)/dt &= n + n + 1 + d\Phi(t)/dt, \\
  & \le c  \bigl( n^o + \sum_{k\in \opt} P({\tilde s}_k)\bigr),
  \end{align*}
  is true for $c = \left(2 + \frac{2}{P^{-1}(1)} \max(1,P(\bar{s}))\right)$, where $c_1 = 2$, $c_2 = 2/P^{-1}(1)$.
  
{Discussion:} By introducing the job-SRPT algorithm, we have effectively removed the precedence constraints within each job. By processing as many distinct jobs, the job-SRPT policy ensures that there is sufficient drift available for the choice of the potential function, and allows leveraging of the recent analysis of the multi-server SRPT algorithm \cite{vaze}. It is worth noting that the analysis of this section (Theorem \ref{thm:srptimproved}) holds even when the {\bf precedence constraints are arbitrary within each job}. 

%% file: WmaxbywminGuarantee.tex
\section{Task Based Flow Time}\label{sec:task}
In this section, we consider the task-by-task flow time where the cost
metric is $C_{\text{task}}$ \eqref{eq:costtask}. We consider the online
setting, where jobs arrive over time arbitrarily, and the objective is
to propose an online algorithm and bound its competitive ratio.  The
competitive ratio guarantee of our proposed online algorithm will be a
function of $\beta = \frac{w_{\max}^t}{w_{reduce, \min}}$, where
$w_{\max}^t$ and $w_{reduce, \min}$ is the maximum size of any
map/reduce task and the minimum size of any reduce task,
respectively. Since for the MapReduce framework, this ratio is
typically small, the guarantees are meaningful
\cite{hammoud2011locality}. Moreover, even when $\beta=1$, the problem
is non-trivial, since how to construct potential function for the case
when the number of map tasks is smaller than the number of servers but
there are large number of caged reduce jobs is not clear, and to the
best of our knowledge no competitive ratio guarantees are known in the
online case.

$\opt$: Throughout this section, we assume that the $\opt$ has {\bf no precedence constraints across tasks}, and can process any of its outstanding tasks at any time. This can only decrease the $\opt$'s cost. Moreover, we assume that without precedence constraints, $\opt$ follows the multi-server SRPT algorithm, i.e., among its all outstanding tasks $n^o$, it processes the $\min\{n^o, K\}$ shortest tasks (map/reduce since there are no precedence constraints for $\opt$) on the $\min\{n^o, K\}$ servers at any time. Since the multi-server SRPT algorithm that has access to future arrivals and can choose arbitrary speed depending on it is $P(2-1/K)$-competitive with respect to $\opt$ with no precedence constraints \cite{vaze}, we will get an additional factor of $P(2-1/K)$ in the final competitive ratio.

We need the following definition to describe our algorithm.
\begin{definition}\label{defn:load} 
For a map task $j$ belonging to job $i$, define its {\it load} as $\ell_{j}=\sfr_i+1$, the number of reduce tasks in the job $i$ that it precedes plus $1$ for itself. 
%For a reduce task $k$ of job $i$ its {\it load} $\ell_k = 1$. 
%Let $b_j(t)$ ($b^o_j(t)$) be the number of outstanding map tasks belonging to the same job as task $j$ with the algorithm ($\opt$)
%at time $t$. For notational simplicity let, ${\hat b}(t) = \min\{b(t),K\}$ and $w({\hat b}, t) = \frac{\ell_j +{\hat b}_j(t)}{{\hat b}_j(t)}$.
%Let job $i$ (that has a map task $j$) arrive at time $t$. Then the set of outstanding map tasks at time $t$ that have arrived before or with map task $j$ be $A_j$. In particular, all other map tasks belonging to job $i$ are also part of $A_j$. Then, 
%For a map task $j$, its rank at time $t$ is defined to be  
%$\text{rank}_j(t) = \sum_{k \in A_j} \ell_k + \ell_j$, where the set of outstanding map tasks at time $t$ that have arrived before or with map task $j$ .
For a map task $j$, $b_j(t)$ ($b^o_j(t)$) be the number of (brother map tasks) outstanding map tasks belonging to the same job as task $j$ with the algorithm ($\opt$)
at time $t$. Moroever, let ${\hat b}(t) = \min\{b(t),K\}$ and $z_j({\hat b}_j, t) = \frac{\ell_j +{\hat b}_j(t)}{{\hat b}_j(t)}$. 
\end{definition}

{\bf Algorithm:} Let the total number of outstanding tasks at time $t$ be $n(t)$ out of which $n_f(t)$ are free, i.e., $n_f(t)$ is the sum of all the map tasks and all the free reduce tasks. The algorithm processes the $\min\{n_f(t), K\}$ shortest free tasks on the $\min\{n_f(t), K\}$ servers. 
If $n_f(t) \ge K$, then all servers  process at speed $P^{-1}\left(\frac{n(t)+1}{K}\right)$. Otherwise,  server $k$ has speed
\vspace{-0.1in}
\begin{align}
s_k(t) = \begin{cases} 
                             P^{-1}\left(\frac{\ell_j +{\hat b}_j(t)+1}{{\hat b}_j(t)}\right) \ \text{if $k$ is processing map task} \  j, & \\
                              P^{-1}(1) \ \quad \quad \text{if $k$ is processing a free reduce task.} &   
\end{cases}
\end{align}
The main result of this section is as follows.

\begin{theorem}\label{thm:online}
The competitive ratio of the proposed online algorithm is $P(2-1/K) (2\beta +2+2(\alpha-1)) +o(1)$, where $\beta = \frac{w_{\max}^t}{w_{reduce, \min}}$.
\end{theorem}
Next, we work towards proving Theorem \ref{thm:online}.
At time $t$, let $n^o(t,q)$ and $n(t,q)$ denote the number of
unfinished tasks (map + reduce) under the $\opt$ and the algorithm, respectively, with remaining size at least $q$. In particular, $n^o(t,0) = n^o(t)$ and
$n(t,0) = n(t)$. Let $d(t,q) = \max\left\{0, n(t,\frac{q}{\beta}) - n^o(t,\frac{q}{\beta})\right\}$. 
We consider the potential function
\begin{align*}\label{defn:phi}
\Phi(t) &= \Phi_1(t)+\Phi_2(t) + \Phi_3(t), \\ 
  \Phi_1(t) &= c_1\int_{0}^\infty f\left(\frac{d(t,q)}{K}\right) dq,\\
\Phi_2(t) & = c_2 \int_{0}^\infty (n(t,q) - n^o(t,q)) dq,
\end{align*}
where $c_1,c_2$ are positive constants and function $\Phi_3$ which is specially designed for handling the precedence constraints is defined in \eqref{defn:new:phi3}.
%The rest of the section is dedicated to proving
%Theorem~\ref{thm:online}, where we assume that $\opt$ is the U-SRPT
%algorithm as discussed earlier.  The unconstrained U-SRPT algorithm
%\cite{vaze} and the speed choice is defined as follows.
%\begin{definition}(U-SRPT \cite{vaze})\label{eq:speeddefusrpt} Let $n^o(t)$ denote the number
%  of unfinished tasks (no MapReduce constraints are enforced) with
%  U-SRPT.  The U-SRPT algorithm maintains a single queue and serves
%  the $\min\{K, n^o(t)\}$ shortest tasks at any time $t$ with identical speed of $P^{-1}\left(\frac{n^o(t)}{K}\right)$ when $n^o(t)\ge K$ and $P^{-1}(1)$ other wise for all servers.
%  \end{definition}
%For a map task $j$, let $q_j$ ($q_j^o$) be the remaining size of a map task $j$ with the algorithm ($\opt$) at time $t$, and recall that 
%$b_j(t)$ ($b^o_j(t)$) be the number of outstanding map tasks belonging to the same job as task $j$ with the algorithm ($\opt$)
%at time $t$. For notational simplicity let, ${\hat b}(t) = \min\{b(t),K\}$ and $w({\hat b}, t) = \frac{\ell +{\hat b}(t)}{{\hat b}(t)}$. Let $M(t)$ ($M^o(t)$) be the number of outstanding map tasks with the algorithm ($\opt$) at time $t$, and for the $j^{th}$ map task $h_j(t,q)=1$ ($h_j^o(t,q)=1$) for $q \le q_j$ ($q \le q_j^o$) 
%and zero otherwise.
Let $m(t) = \sum_{i \in \text{outstanding jobs}} m_i(t)$ be the number of outstanding map tasks with the algorithm at time $t$. Similarly, define $m^o(t)$ for the $\opt$. For the $j^{th}, j=1, \dots, m(t) (m_o(t))$ map task, let $q_j$ ($q_j^o$) be its remaining size with the algorithm and the $\opt$, respectively. Then for the $j^{th}$ map task, let $h_j(t,q)=1$ ($h_j^o(t,q)=1$) for $q \le q_j$ ($q \le q_j^o$) 
and zero otherwise. Define 
\begin{eqnarray} \nn
 \Phi_3(t) &=& c_3 \left( \sum_{j \in m(t)} \int g_{z_j({\hat b}_j, t)} \left( z_j({\hat b}_j, t) h_j(t, q)\right) dq \right.   \\  \label{defn:new:phi3}
&& - \left. \sum_{j \in m_o(t)} \int g_{z_j^o({\hat b}^o_j, t)} \left( z_j^o({\hat b}^o_j, t) h_j^o(t, q)\right)dq\right),
\end{eqnarray}
where for $i \in \bbN,$ let
$g_a\left(ai \right)-g_a\left(a(i-1)\right) =
\Delta\left(a i \right) :=
P'\left(P^{-1}\left(a i \right)\right)$ and $g_a(0) = 0$. 

Note that both the boundary conditions are satisfied for $\Phi_1$ and $\Phi_2$.
However,  $\Phi_3$ only satisfies the first boundary condition, i.e., $\Phi_3=0$ before any job is
  released and after all jobs are finished. 
It does not satisfy the second boundary condition, however, the total jump in $\Phi_3$ at points of discontinuity can be bounded as follows.

\begin{lemma}\label{lem:jump3} 
The total upward jump in
  the potential function $\Phi_3(t)$ is at most 
  $O(\sfJ w_{\max} K^2  r_{\max}^{1-\frac{1}{\alpha}})$ that is $ o(\sum_{j \in \cJ} (\sfm_{jk} + \sfr_{j\ell}))$ under Assumption \ref{ass:size}.
\end{lemma} 

%  \begin{lemma}
%  \label{lemma:f_prop2}
%  For $a,d \in \bbN$  where $d \leq a,$ $$g_{a}\left(\frac{d}{a}\right) \leq g_{a+1}\left(\frac{d+1}{a+1}\right).$$
%
%  Moreover, for $a,d \in \bbN,$  $$g_a\left(\frac{d}{a}\right) \geq g_{a+1}\left(\frac{d}{a+1}\right).$$ 
%\end{lemma}
%\begin{proof}
%  To prove the first statement, we note that
%\begin{align*}
%g_{a}\left(\frac{d}{a}\right) &= \sum_{j = 1}^d \Delta(j/a) \\
%& \stackrel{(a)}{\leq} \sum_{j = 1}^d \Delta(j+1/a+1) = \sum_{j = 2}^{d+1} \Delta(j/a+1) \\
%& \leq \sum_{j = 1}^{d+1} \Delta(j/a+1) = g_{a+1}\left(\frac{d+1}{a+1}\right).
%\end{align*}
%Here, $(a)$ follows from the monotonicity of $\Delta(\cdot).$ The second statement of the lemma is trivial: 
%$$g_a\left(\frac{d}{a}\right) = \sum_{j = 1}^d \Delta(j/a) \geq \sum_{j = 1}^d \Delta(j/a+1) = g_{a+1}\left(\frac{d}{a+1}\right).$$
%\end{proof}

Next, we work towards finding the drift (derivative) of $\Phi(t)$. 
The drift terms for the potential function $\Phi_1$ can be derived from \cite[Lemma 10]{vaze}, where the only difference between $\Phi_1$ considered in this paper and \cite{vaze} is that $q$ scaled by $1/\beta$. In \cite{vaze}, without any precedence constraints, the shortest $\min\{n(t),K\}$ jobs (that have a single task) with size $q_1\le  \dots \le q_{\min\{n(t),K\}}$ are processed by the algorithm, making 
$n(t, q_i) = n(t)-i$ for all $i=1, \dots, \min\{n(t),K\}$.
With the precedence constraint in this section, the  algorithm processes the $\min\{n_f, K\}$ {\bf free} tasks which are not necessarily the $\min\{n, K\}$ shortest tasks, since a caged reduce task could be smaller but not being processed. 
Let $q_1 \le q_2, \dots, \le q_{\min\{n_f, K\}}$ be the remaining sizes of the $\min\{n_f, K\}$ shortest free tasks being processed.
This scaling of $q$ by $\beta$ ensures that even when the absolute shortest task is not being processed by the algorithm $n(t, q_i/\beta) \ge n(t)-i$ for $i=1,\dots,\min\{n_f, K\}$, since the size of any caged reduce job is at least $w_{reduce, \min}$. 
Since $\beta$ is a constant, scaling of $q$ by $1/\beta$ translates to having an additional $\beta$ factor in the derivative of $\Phi_1(t)$, and we get the following result, where ${\tilde s}_k$ is the speed of server $k$ with the $\opt$.
\begin{lemma}\cite[Lemma 10]{vaze}
  \label{lemma:phi1_OPT-SRPT}
  For $n_f \geq K$,$ d\Phi_1/dt \le  \frac{c_1}{\beta} n^o - \frac{c_1}{\beta} n + \frac{c_1}{\beta} \left(\frac{K-1}{2}\right) + c_1 \sum_{k \in \opt } P(\tilde{s}_k),$ 
  while for $n_f < K,$
  $d\Phi_1/dt \le \frac{c_1}{\beta} n^o - \frac{c_1}{\beta} \frac{n(n+1)}{2K} + c_1 \sum_{k \in \opt} P(\tilde{s}_k)$.
\end{lemma}
\begin{lemma} \label{lemma:phi2_OPT-SRPT}\cite[Lemma 11]{vaze}    
$d\Phi_2/dt \leq -c_2 \min(K,n) P^{-1}(1)+  c_2\sum_{k \in \opt}  P(\tilde{s}_k)$.
\end{lemma}

Next, we quantify the more {\bf non-trivial} drift of $\Phi_3(t)$ because of the processing by the algorithm and the $\opt$.
We will need the following result from \cite{bansal2009speedconf}.
\begin{lemma}\label{lem:bansal}[Lemma 3.1 in~\cite{bansal2009speedconf}]
 Let $P$ satisfy Assumption \ref{ass:Pconvex}.  Then for $s_k, {\tilde s}_k, x \ge 0$,
\begin{align*}
\Delta(x)(-s_k + {\tilde s}_k) \le& \left(-s_k +P^{-1}(x)\right)\Delta(x) 
%\label{eq:powerbound}
 + P({\tilde s}_k) -x.
\end{align*}
\end{lemma}
Let $n_{fr}(t)$ be the free reduce tasks at time $t$, and let ${\hat n}(t)= n(t) - n_{fr}(t)$ $({\hat n}^o(t))$ be the number of outstanding tasks other than the free reduce tasks at time $t$ with the algorithm ($\opt$). The {\bf main new result} we need with precedence constraints is as follows, where $\b1_{\{e\}} =1$ if event $e$ is true and zero otherwise.

\begin{lemma}\label{lem:drfitspecial} 
  \label{lemma:phi3drift}
 $d\Phi_3/dt   \le - c_3 \b1_{\{n_f< K\}} {\hat n}(t) 
     + c_3(\alpha-1) ({\hat n}^o(t))+ c_3 \sum_{k =1}^{m^o(t)} P(\tilde{s}_k)$,
\end{lemma}
\begin{proof}
For notational simplicity, we suppress the time index, since we are considering a fixed time $t$.
  Consider the change in $\Phi_3$ due to $\opt$ ($h_j^o(q) \rightarrow
  h_j^o(q)-1$) for at most $m^o$ map tasks in the interval $[q_j^o, q_j^o-{\tilde s}_jdt]$ where $q_j^o$ is the remaining size of the $j^{th}$ map task with the $\opt$: $ $
  \begin{eqnarray}\nn
    d\Phi_3 &\le & c_3 \sum_{j= 1}^{m^o}
    \biggl[ g_{z_j^o({\hat b}^o)} \left( z_j^o({\hat b}^o) h_j^o( q_j^o)\right)dq \biggr. \\ \nn
    & &   -
     \biggl. g_{z_j^o({\hat b}^o, t)} \left( z_j^o({\hat b}^o) (h_j^o(q_j^o)-1)\right)dq \biggr]{\tilde s}_j dt,  \\ \nn
     &= & 
     c_3\sum_{j= 1}^{m^o}\Delta \left(\frac{\ell_j^o+{\hat b}_j^o }{{\hat b}_j^o}h_j^o(q_j^o)\right){\tilde s}_j dt, \\ \label{eq:dummy766}
     &=&  
     c_3\sum_{j= 1}^{m^o}\Delta \left(\frac{\ell_j^o+{\hat b}_j^o }{{\hat b}_j^o}\right){\tilde s}_j dt, 
     \end{eqnarray}
     since $h_j^o(q_j^o)=n_j^o=1$.
    Using Lemma \ref{lem:bansal} individually on the ${m^o}$ terms of \eqref{eq:dummy766} with $s_k=0$, and noting that  $P^{-1}(x) \Delta(x) = (\alpha-1)x$, we get 
    \begin{eqnarray} \label{eq:dummy767}
    d\Phi_3 & \le & c_3 (\alpha-1) \sum_{j= 1}^{m^o} \left(\frac{\ell_j^o+{\hat b}_j^o}{{\hat b}_j^o}\right) + c_3 \sum_{j =1}^{m^o} P(\tilde{s}_k), \\
   & \le &c_3  (\alpha-1)  ({\hat n}^o(t)) + c_3 \sum_{j =1}^{m^o} P(\tilde{s}_k),
    \end{eqnarray}
    since the number of tasks $n_k$ in job $k$ to which map task $j$ belongs $n_k = \ell_j^o+{\hat b}_j^o$.

Next, we consider the change in $\Phi_3$ because of algorithm's execution only when $n_f < K$ i.e., there are $n_m\le K-1$ outstanding map tasks. Let 
these $n_m$ outstanding map tasks belong to $|J|$ distinct jobs. Because of algorithm's execution, $h_j(q)
\rightarrow h_j(q)-1$ for map task $j$ in the interval $[q_j, q_j-s_j dt]$ where $q_j$ is the remaining size of the 
$j^{th}$ map task with the algorithm. Thus, 
\begin{eqnarray}\nn
  d\Phi_3 &=& c_3 \sum_{j = 1}^{n_m}
  \biggl[g_{z_j({\hat b})} \left( z_j({\hat b}) (h_j(q_j)-1)\right) \biggr. \\ \nn
  &  \biggl. -
    g_{z_j({\hat b}, t)} \left( z_j({\hat b}) (h_j(q_j))\right) \biggr]s_j dt, \\  \nn
 &= & 
     -c_3\sum_{j= 1}^{n_m}\Delta \left(\frac{\ell_j+{\hat b}_j}{{\hat b}_j}h_j(q_j)\right)s_j dt,      \\ \label{eq:dummy866}
      &= &   -c_3\sum_{j= 1}^{n_m}\Delta \left(\frac{\ell_j+{\hat b}_j}{{\hat b}_j}\right) s_j dt,
      \end{eqnarray}
     since $h_j(t, q_j)=1$.
    Using Lemma \ref{lem:bansal} individually on each of the $n_m$ terms with ${\tilde s}_k=0$, we get     
    \begin{eqnarray}\nn
 d\Phi_3 &\le & c_3 \sum_{j = 1}^{n_m}
    \left(-s_j +P^{-1}\left(\frac{\ell_j+{\hat b}_j}{{\hat b}_j}\right)\right)\Delta\left(\frac{\ell_j+{\hat b}_j }{{\hat b}_j}\right)\\  \nn
    & & \quad - c_3 \sum_{j = 1}^{n_m}\left(\frac{\ell_j+{\hat b}_j }{{\hat b}_j}\right) \\ \nn
&\stackrel{(a)}=&  - c_3 \sum_{j = 1}^{n_m}\left(\frac{\ell_j+{\hat b}_j }{{\hat b}_j}\right), \\ \nn
    & \stackrel{(b)}=&  - c_3 \sum_{k=1}^{|J|}\sum_{j = 1, j\in J_k}^{b_j}\left(\frac{\ell_j+{\hat b}_j }{{\hat b}_j}\right), \\ \label{eq:dummy867}
    &=& -c_3  {\hat n}(t),
 \end{eqnarray}   
 where $(a)$ follows since the speed of map task $j$ is 
    $s_j = P^{-1}\left(\frac{\ell_j +{\hat b}_j(t)+1}{{\hat b}_j(t)}\right)$ making the first term $<0$, and in $(b)$ we have separated the contributions from the $|J|$ outstanding jobs, and $J_k$ represents the $k^{th}$ job which has $b_j+\ell_j$ tasks for $j\in J_k$, and the final inequality follows since the total number of map tasks is $\sum_{j\in n_m} b_j \le K-1$. Combining \eqref{eq:dummy767} and \eqref{eq:dummy867}, we get the result.
 \end{proof}
 
% Next, we consider the drift  term $d\Phi_1/dt$. Recall Assumption \ref{ass:size} that states that size of each map task is at most the size of any reduce task. Thus, when  $n_f \geq K$, even under the MapReduce precedence constraints, the $K$ tasks being processed by the MR-SRPT are precisely the $K$ shortest tasks available with the MR-SRPT algorithm. Thus, the change in $\Phi_1$ because of execution of algorithm MR-SRPT ($n(q)
%\rightarrow n(q)-1$ for $q=q(1), \cdots, q(K)$): and $n(q(i)) = n-i$ for $i=1,2,\dots, K$. Moreover, since the $U-SRPT$ is unconstrained, the change in $\Phi_1$ because of execution of algorithm U-SRPT ($n^o(q)
%\rightarrow n^o(q)-1$ for $q=q(1), \cdots, q(K)$): and $n^o(q(i)) = n-i$ for $i=1,2,\dots, K$. Thus, we can directly use the following result from \cite{vaze} when  $n_f \geq K$.
%\vspace{-0.1in}
Next, we want to show that \eqref{eq:mothereq} holds for an appropriate choice of $c$, using the bounds that we have derived for $d\Phi_1/dt, d\Phi_2/dt$ and $d\Phi_3/dt$. 
We begin with the case that $n_f \geq K$. In this case, we only count the positive terms (because of $\opt$) of $d\Phi_3/dt$ of Lemma \ref{lem:drfitspecial} and ignore the negative contribution because of the execution of the algorithm.  Combining Lemma \ref{lemma:phi1_OPT-SRPT}, Lemma \ref{lemma:phi2_OPT-SRPT}, and Lemma \ref{lem:drfitspecial}, we get $ n + \sum_{k=1}^KP(s_k) + d\Phi_1/dt + d\Phi_2/dt + d\Phi_3/dt $
\begin{align}
& \le  n+ n + \frac{c_1}{\beta} n^o - \frac{c_1}{\beta} n + \frac{c_1}{\beta} \left(\frac{K-1}{2}\right) + c_1 \sum_{k \in \opt} P({\tilde s}_k) \\ \label{dummy127}
&\quad  -c_2 K P^{-1}(1) +  c_2\sum_{k \in \opt} P({\tilde s}_k)  + c_3 (\alpha-1) n^o, \\
&\quad + c_3 \sum_{k \in \opt} P({\tilde s}_k) ,
\end{align}

When $n_f<K$ and $n_m > 0$,  combining Lemma \ref{lemma:phi1_OPT-SRPT}, Lemma \ref{lemma:phi2_OPT-SRPT}, and Lemma \ref{lem:drfitspecial}, $n + \sum_{k=1}^KP(s_k) + d\Phi_1/dt + d\Phi_2/dt + d\Phi_3/dt$ 
\begin{align}\nn
&\le n + {\hat n} + n_{fr} P(P^{-1}(1)) + \frac{c_1}{\beta} n^o - \frac{c_1}{\beta} \frac{n(n+1)}{2K} \\ \nn
& \quad  + c_1 \sum_{k \in \opt} P({\tilde s}_k) -c_2 n_f P^{-1}(1)  \\ \label{dummy128}
&\quad + c_2\sum_{k \in \opt} P({\tilde s}_k) -c_3  {\hat n} + c_3 (\alpha-1)n^o + c_3\sum_{k \in \opt} P({\tilde s}_k).
\end{align}
%When $n_f<K$, and $n_m > 0$, this means that there are less than $K-n_m$ free reduce tasks, and 
%clearly, $\sum_{i=1}^{N_\sfm} \ell_i \ge n - n_{fr}$. 
Thus, for both cases \eqref{dummy127} and \eqref{dummy128}, we get $
 n + \sum_{k=1}^KP(s_k) + d\Phi_1/dt + d\Phi_2/dt + d\Phi_3/dt  \le c \left(n^o + \sum_{k \in \opt} P({\tilde s}_k)\right)$,
for $c = c_1+c_2+(\alpha-1)c_3$ for $c_3 \ge 2$.
The final case is when the algorithm has $n_f<K$ and $n_m =0$, i.e., there are less than $K$ free tasks and all of them are of reduce type, i.e., $n = n_{fr}$, for which 
the proof is provided in Appendix \ref{app:finalcasetaskbased}.
Thus, $c$ for which \eqref{eq:mothereq} holds for all the three cases is $c = 2\beta +2+2(\alpha-1)$. 

Moreover, recall that that $C_{\text{task}}^\opt\ge  P(s^\star) (\sum_{j \in \cJ} (\sfm_{jk} + \sfr_{j\ell}))$ from Proposition \ref{prop:optlb}. Thus, from Lemma 
\ref{lem:jump3} and Remark \ref{rem:boundary}, we have $\sfD = o(1)$, thus, completing the proof of Theorem \ref{thm:online}.

{\it Discussion:} In this section, we presented a simple speed scaling algorithm that processes as many shortest free tasks as possible at speed such that the total power used is equal to the total number of outstanding tasks. Without precedence constraints, this speed choice can be shown to be locally optimal, i.e., optimal if no further jobs arrive in future. With precedence constraints, however, making a similar claim is difficult. The derived competitive ratio guarantee depends on the size of the map and reduce tasks, since essentially, we are restricting ourselves to processing shortest free tasks, in which case a shorter reduce task is not being processed while a large (free) map/reduce is being processed. 
%For the MapReduce framework,  
%$\beta$ is generally small, and thus the derived guarantee is still meaningful. Moreover, as pointed out earlier, even when all tasks are of equal size, finding an algorithm with a non-trivial upper bound on the competitive ratio in the presence of precedence constraints is a challenge, and was not known to the best of our knowledge. 
%On the theoretical aspect, we believe that philosophy of processing as many free jobs with the shortest size is correct, it is the analysis that is loose.

%% file: NewProof.tex
%\subsection{Online Unsplittable map tasks Setting}\label{sec:onlineunsplit}

\section{Task Based Flow Time {$P(s) = s^\alpha$}, $\alpha < 2$}

In the previous section, we derived competitive ratio guarantees as a function of $\beta$ and $\alpha$.
In this section, we show that we can get a competitive ratio independent of the instance of the problem ($\beta$) that only depends on the value of $\alpha$ as long as $\alpha <2$.

%In the previous section, we considered the setting where all jobs are available at time $0$, which allows the algorithm 
%to complete all the map tasks before starting to work on the reduce tasks. 
%In the online setting, the biggest challenge is the following condition, where there are $n_m< K$ map tasks but lot of reduce tasks that are caged and $n_m$ keeps changing with time. When all jobs are available at time $0$, we took care of this by injecting extra map tasks to keep $n_m = K$ until all map tasks were finished. In the online case, that is not possible since the increase in potential function doing so will become a function of the total number of arriving jobs. Thus, the  online case is far more challenging and one needs to take care of the changing $n_m$ appropriately in the potential function. 

$\opt$: Similar to previous section, we assume that the $\opt$ has {\bf no precedence constraints across tasks}, and can process any of its outstanding tasks at any time. This can only decrease the $\opt$'s cost. With a single server and without precedence constraints, $\opt$ always processes the shortest task which could be a map/reduce task. Unlike the 
previous section, however, we do not restrict the $\opt$ to follow the SRPT algorithm without precedence constraints and avoid the $P(2-1/K)$ penalty as in Theorem \ref{thm:online}.

%To derive competitive ratio bound for the proposed MR-SRPT algorithm
%that respects the MapReduce constraints with speed choice (Definition
%\ref{defn:onlinealgo}), we compare it against the unconstrained $\opt$
%called $U-\opt$ that {\bf does not} have to follow the MapReduce
%constraints. In the unconstrained setup, for any algorithm, on a job
%arrival, all its map and reduce tasks are free, and can be executed in
%any order. Clearly, the cost $C_{U-\opt}$ of the unconstrained
%$U-\opt$, $$C_{U-\opt} \le C_\opt^{MR},$$ where $C_\opt^{MR}$ is the
%cost of the $\opt$ that respects the MapReduce constraints.  Next, we
%use the fact that under the unconstrained tasks setup, the SRPT speed
%scaling algorithm \eqref{eq:speeddefusrpt} proposed in \cite{vaze},
%called U-SRPT hereafter, is constant ($\kappa_2$) competitive with
%respective to the unconstrained $\opt$, and certain properties of
%U-SRPT to derive competitive ratio bound for MR-SRPT with respect to
%U-SRPT.  Thus, suppose if we can show that
%$$C_{MR-SRPT}(\sigma) \le \kappa_1 C_{U-SRPT}(\sigma)$$ for any input
%$(\sigma)$, then since
%$$C_{U-SRPT}(\sigma)\le \kappa_2 C_{U-\opt}(\sigma)\le C_\opt^{MR}(\sigma),$$
%from \cite{vaze}, we get that
%$$C_{MR-SRPT}(\sigma)
%\le \kappa_1 \kappa_2 C_\opt^{MR}.$$ From \cite{vaze}, we know that $$\kappa_2 = 4 P(2-1/K).$$ Thus,  to get that the competitive ratio of the proposed algorithm is at most $\kappa_1\kappa_2$, in the rest of the section, we find $\kappa_1$ such that $$C_{MR-SRPT}(\sigma) \le \kappa_1 C_{U-SRPT}(\sigma)$$ for any input $(\sigma)$.
%  
\input{NewProofPhi3.tex}

%% file: NewProofPhi3.tex
%\section{$\Phi_3$}
Without loss of generality, we will assume hereafter that the number of reduce tasks (load of any map task) in each job are identical, i.e., $\ell_j = \sfr$ for all map tasks $j$. 
This can be done, since if they are unequal ($\ell_j = \sfr_j$ in job $j$), and the maximum number of reduce tasks with any job is $\sfr_{\max}$, then we can add $\sfr_{\max} - \sfr_j$ tasks of zero size, with equal change to the cost of the algorithm or the $\opt$.

%\begin{definition}\label{defn:load} 
%For a map task $j$ belonging to job $i$, define its {\it load} as $\ell_{j}=\sfr+1$ to be the number of reduce tasks in the job $i$ that it precedes plus $1$ for itself. 
%For a reduce task $k$ of job $i$ its {\it load} $\ell_k = 1$. 
%Let $b_j(t)$ ($b^o_j(t)$) be the number of outstanding map tasks belonging to the same job as task $j$ with the algorithm ($\opt$)
%at time $t$. For notational simplicity let, ${\hat b}(t) = \min\{b(t),K\}$ and $w({\hat b}, t) = \frac{\ell_j +{\hat b}_j(t)}{{\hat b}_j(t)}$.
%Let job $i$ (that has a map task $j$) arrive at time $t$. Then the set of outstanding map tasks at time $t$ that have arrived before or with map task $j$ be $A_j$. In particular, all other map tasks belonging to job $i$ are also part of $A_j$. Then, 
%For a map task $j$, its rank at time $t$ is defined to be  
%$\text{rank}_j(t) = \sum_{k \in A_j} \ell_k + \ell_j$, where the set of outstanding map tasks at time $t$ that have arrived before or with map task $j$ .
%\end{definition}
%\begin{definition}We define a reduce task belonging to job $i$ to be {\it free} at time $t$ if all the map tasks of job $i$ have been finished before time $t$. A reduce task that is not free is called {\it caged}.
%\end{definition}

\subsection{Single Server Case}\label{subsec:ss}
First we consider the single server case, since this itself poses new challenges with the precedence constraints. The {\bf cumulative (remaining) map size} of 
a job is defined as the sum of the sizes of all the (remaining) map tasks in a job.

{\bf Algorithm:} We propose the following algorithm, that defines which tasks should be executed and at what speed. At time $t$, let the number of outstanding map tasks (summed across different jobs) be $m(t)$, the number of free reduce tasks be $r_f(t)$, and the number of jobs with at least one unfinished map task be $J(t)$.  
The algorithm will process (upto) three tasks simultaneously. 
\begin{itemize}
\item $M_1$: The shortest remaining map task among all the outstanding map tasks with the algorithm at speed $s_m(t) = P^{-1}(m(t)+1)$.
\item $M_2$: The shortest remaining map task of the job that has the smallest remaining cumulative map size among all the outstanding jobs at speed $s_J(t) = P^{-1}(\sfr J(t)+1)$.
\item $R_f$: The shortest remaining free reduce task among all the outstanding free reduce tasks at speed $s_f(t) = P^{-1}(r_f(t)+1)$.
\end{itemize}
Thus, the total speed is $s(t) = s_m(t) + s_J(t) + s_f(t)$.
Note that $M_1$ could be equal to $M_2$ in which case only two tasks are processed simultaneously. 

For the algorithm, let $m(q)$, $J(q)$, and $r(q)$ denote the number of
unfinished map tasks, jobs with at least one unfinished map tasks, and the number of unfinished reduce tasks (free or caged) with
remaining cumulative size of at least $q$, respectively. The quantities for the $\opt$ are denoted with a superscript $o$.
Let
$d_m(q) = \max\left\{0,m(q) - m^o(q)\right\}$, $d_J(q) = \max\left\{0,\sfr (J(q) - J^o(q))\right\}$, and $d_r(q) = \max\left\{0,r(q) - r^o(q)\right\}$.
We consider the potential function
\vspace{-0.1in}
\begin{align}\nn
\Phi(t) = \Phi_1(t)+\Phi_2(t)+\Phi_3(t), \\ \nn
\  \Phi_1(t) = c_1\int_{0}^\infty f\left(d_m(q)\right) dq, \\ \label{defn:phiss}
\  \Phi_2(t) = c_2\int_{0}^\infty f\left(d_J(q)\right) dq,
\end{align}
and $ \Phi_3(t) = c_3\int_{0}^\infty f\left(d_r(q)\right) dq,
$ where $c_1,c_2,c_3$ are positive constants to be determined later.

The motivation to consider processing three tasks simultaneously is closely tied with the choice of potential function $\Phi(t)$, where we need the total drift (derivative of the potential function $\Phi$) to decrease at a rate $- c n(t)$, where $n(t)$ is the total number of outstanding tasks with the algorithm. Because of the precedence constraints, the algorithm is not necessarily allowed to process the shortest task (which can be a caged reduce task), in contrast to the no-precedence constraint case as in \cite{bansal2009speedconf}. Thus, to get sufficient drift, we combine the drifts from processing three jobs, $M_1$ gives $-m(t)$, $M_2$ gives $-\sfr J(t)$ and $R$ gives $-r_f(t)$. Since $m(t) + \sfr J(t) + r_f(t) \ge n(t)$, we get sufficient drift.

%\begin{definition}\label{defn:onlinealgo}[MR-SRPT](Algorithm and Speed Choice)
%  Let $n(t)= n_f(t) + n_c(t)$, and $n_f(t)$ be the number of outstanding tasks and number of free tasks at
%  time $t$, respectively. With multiple servers, the shortest remaining processing
%  time (SRPT) algorithm maintains a single queue and serves the
%  $\min\{K, n_f(t)\}$ shortest tasks at any time $t.$
%The speed for server $k$  with the SRPT algorithm is chosen as follows. When $n_f \ge K$, for all servers $k=1,\dots, K$, 
%$$s_k(t) = P^{-1}\left(\frac{n(t)}{K}\right).$$
%When $n_f< K$, then each of the $n_{fr}$  free reduce tasks are run at speed $P^{-1}(1)$, while for the $k^{th}$ server executing the $j^{th}$ map task, the speed is $s_k(t) = P^{-1}\left(\frac{\ell_j +{\hat b}_j(t)+1}{{\hat b}_j(t)}\right)$, where $\ell_j$ is the load of the number of reduce tasks belonging the same job as the $j^{th}$ map task.
%\end{definition}
%Thus, if the number of free tasks is at least as much as the number of servers, then each server runs at the same speed of $P^{-1}\left(\frac{n(t)}{K}\right)$. Otherwise,  map and reduce tasks run at  different speeds. 
%

Following this approach, the main result of this section is as follows. 
\begin{theorem}\label{thm:onliness} 
For the single server case, for $\alpha < 2$, the  proposed algorithm has a competitive ratio of $8+\frac{4}{2-\alpha}+3^{\alpha}$.
   \end{theorem}
   The main ingredients in proving Theorem \ref{thm:onliness} are the following three Lemmas (proofs in Appendix), where ${\tilde s}(t)$ is the speed of 
   server with the $\opt$ at time $t$.
 \begin{lemma}\label{lem:driftphi1}
 $d\Phi_1/dt \le c_1P({\tilde s}(t)) - c_1(m(t) - m^o(t))$.
 \end{lemma}

 \begin{lemma}\label{lem:driftphi2}
 $d\Phi_2/dt \le c_2 P({\tilde s}(t)) - c_2(\alpha-2) \sfr J(t)  + c_2(2-\alpha) \sfr J^o(t)$.
 \end{lemma}   
 \begin{lemma}\label{lem:driftphi3}
$ d\Phi_3/dt \le c_3 P({\tilde s}(t)) - c_3(r_f(t) - (\alpha-1)r(t) - 
 (2-\alpha)r^o(t))$.
 \end{lemma}
 
 Next, we prove that \eqref{eq:mothereq} holds for some $c$ to prove Theorem \ref{thm:onliness} using Lemma \ref{lem:driftphi1}, \ref{lem:driftphi2}, \ref{lem:driftphi3}.
 
 \begin{proof}[ Proof of Theorem \ref{thm:onliness}]
 We only describe the computation when $m(t) \ge m^o(t), J(t) \ge J^o(t)$ and $r(t) \ge r^o(t)$, since otherwise result follows easily.
 Recall that the total number of outstanding jobs with the algorithm is  $m(t) + \sfr J(t) + r_f(t)$, and the total power used by the algorithm is $P(s(t)) = $
 \begin{align*}
 &P(P^{-1}(m(t)+1)+P^{-1}(\sfr J(t)+1) + P^{-1}(r_f(t)+1))  \\ 
 & \quad \le 3^{\alpha-1} (m(t) + \sfr J(t) + r_f(t)+ 3),
 \end{align*} using Lemma \ref{lem:lpnorm} since $P^{-1}(x) = x^{1/\alpha}$ with $n=3$ and $\alpha \ge 1 $. Thus, the running cost $n(t) + P(s(t))$ for the algorithm is  at most $(3^{\alpha-1}+1)(m(t) + \sfr J(t) + r_f(t)) + 3^{\alpha}$.
 Combining Lemma \ref{lem:driftphi1}, \ref{lem:driftphi2}, \ref{lem:driftphi3}, we can write \eqref{eq:mothereq}, $  n(t) + P(s(t)) + d\Phi/dt$\vspace{-0.1in}
 \begin{align*}
  &  \le  (3^{\alpha-1}+1)(m(t) + \sfr J(t) + r_f(t)) + 3^{\alpha} \\ 
  & \quad + (c_1+c_2+c_3) P({\tilde s}(t)) - c_1(m(t) - m^o(t))\\
 & \quad - c_2(\alpha-2) \sfr J(t)  + c_2(2-\alpha) \sfr J^o(t) \\
 & \quad  - c_3(r_f(t) - (\alpha-1)r(t) - 
 (2-\alpha)r^o(t)), \\
 & \le c(m^o(t) + \sfr J^o(t) + r^o(t)+ P({\tilde s}(t))) \le c(n_o(t) + P({\tilde s}(t))),
 \end{align*}
 $c= c_1+c_2+c_3+3^{\alpha}$, where $c_1 = 4, c_2 = \frac{4}{2-\alpha}, c_3 =4$, since $\alpha<2$. Thus, the competitive ratio of the proposed online algorithm is  $8+\frac{4}{2-\alpha}+3^{\alpha}$.
\end{proof}
  \input{multiserveralpha2}

%Note that assuming $\sfn_m = o(\sfn_r)$ is really not a restriction, since it applies widely in practice.

%{\it Discussion:} 
%Thus, we have showed that even with MapReduce precedence constraints, the proposed MR-SRPT algorithm is constant competitive (only depends on $P$) even with respect to the $\opt$ that does not respect the precedence constraints. In comparison to prior work, this is significant, since the prior results were only known for either the offline case with MapReduce constraints or for general DAG precedence constraints when all jobs were available at time $0$ where the approximation ratio scales with $K$.

%The proposed algorithm is easy to implement since it keeps a single
%queue and follows the the SRPT discipline, even though it requires
%migration. With precedence constraints, however, precluding migration
%is cumbersome, since reduce tasks cannot be frozen to run on a fixed
%server, since the MapReduce constraint essentially entails a priority
%on map tasks.

%% file: multiserveralpha2.tex
\subsection{Multiple Server Case}\label{sec:multiserveralpha2}
In this section, we consider the multi-server version of the problem, where there are $K$ homogenous servers each with identical power function $P$. Except this change, everything else remains the same about the model as described in Section \ref{subsec:ss}.
In the presence of precedence constraints, speed scaling with multiple servers pose the following new challenge. Let the number of map plus free reduce jobs be less than $K$, while the number of caged reduce jobs be large. In this case, even though there are a large number of outstanding jobs, some of the servers have to idle, making the construction of the potential function a difficult task. 
%Recall that the choice of potential function should be such that drift is some negative constant times the total number of outstanding tasks.

{\bf Algorithm:} With multiple servers, we propose the following algorithm, that defines which tasks should be executed and at what speed. At time $t$, let the number of outstanding map tasks (summed across different jobs) be $m(t)$, the number of free reduce tasks be $r_f(t)$, and the number of jobs with at least one unfinished map task be $J(t)$.  
\begin{itemize}
\item $M_1$'s: The $\min\{m(t), K\}$ shortest remaining map tasks among all the outstanding map tasks available with the algorithm at speed $s_m(t) = P^{-1}\left(\min\left\{\frac{m(t)+1}{K}, 1\right\}\right)$ are processed on any $\min\{m(t), K\}$ servers.
\item $M_2$'s: The $\min\{J(t), K\}$ jobs that has the smallest remaining cumulative size of all the outstanding map tasks at speed $s_{cm}(t) = P^{-1}\left(\min\left\{\frac{ \sfr J(t)+1}{K}, \sfr\right\}\right)$  are processed on any $\min\{m(t), K\}$ servers, where the shortest map tasks of each job is processed.
\item $R_f$'s: The $\min\{m(t), K\}$ shortest remaining free reduce tasks among all the free reduce tasks at speed $s_f(t) = P^{-1}\left(min\left\{\frac{ r_f(t)+1}{K}, 1\right\}\right)$ are processed on any $\min\{m(t), K\}$ servers.
\end{itemize}
Thus, the algorithm will process (upto) three tasks on each of the $K$ servers simultaneously, where the speed of any server is $s(t) = s_m(t) + s_{cm}(t) + s_f(t)$. For any server, $M_1$ could be equal to $M_2$ for some jobs in which case only two tasks are processed on each server simultaneously.   
  
  \begin{theorem}\label{thm:onlinems} 
For the multi-server case, for $\alpha < 2$, the  proposed algorithm has a competitive ratio of $8+\frac{4}{2-\alpha}+3^{\alpha}$.
   \end{theorem}

 Proof is provided in the Appendix \ref{app:onlinems}, where the novelty over the single server case is the drift $d\Phi_4(t)/dt$ that is needed only when $K >1$. This 
 distinction is required since when number of jobs is less than $K$, but there are a large number of caged reduce jobs, all jobs of $J(t)$ are being processed by the algorithm, thus giving the negative drift of $- \sum_{j\in J(t)} \Delta(\sfr) s_j(t)$ where $s_j(t)$ is the speed at which each of the jobs's shortest map task is being executed.

%% file: Simulations.tex
\section{Simulations}
In this section, we present simulation results for both the job-based flow time + energy and task-based flow time + energy problem. Since the $\opt$ is unknown, for benchmarking the performance, we remove the MapReduce constraints, and use the same speed scaling choice as that of the algorithms presented in the paper. For the task-based flow time + energy problem, without the MapReduce constraints the algorithm is the same the multi-server SRPT algorithm with speed choice prescribed in \cite{vaze}, that is constant competitive with respect to the $\opt$. Thus, our benchmark is meaningful. For the job-based flow time + energy problem, without the MapReduce constraints, there is no guarantee known on the Job-SRPT algorithm as far as we know, or for that matter any algorithm. 

For all the plots, we use $\alpha=2.5$, and number of servers $K=5$.
In the simulation setup, we consider a slotted time where in each slot, the number of jobs arriving is Poisson distributed with mean $\lambda$. For each job, the number of map and reduce tasks are Poisson distributed  with mean $2$ and $2$ (reasonable with $K=5$), and the size of a map job and a reduce job is exponentially distributed with mean $3$ and $5$, respectively. For each iteration, we generate jobs for $1000$ slots, and count its flow time + energy, and iterate over 1000  iterations. The performance of the two algorithms is compared for the same realizations of the random variables, and then averaged over iterations. 

In Fig. \ref{fig:job}, we plot the comparison of the job-based flow time + energy with and without MapReduce constraints as a function of $\lambda$, where even though in theory we consider that all jobs are available at time $0$, in simulation we let jobs arrive over time. Clearly, the ratio of the average costs is at most $5$. Moreover, from simulations we observed  that even the maximum ratio of the two costs remain below $5$. 
Similarly, in Fig. \ref{fig:task}, we plot the comparison of the task-based flow time + energy with and without MapReduce constraints as a function of $\lambda$, and observe similar behaviour.

Recall that the competitive ratio bound for the  proposed algorithm for the task-based flow time + energy problem is a function of the ratio of the largest map/reduce task and the smallest reduce task $\beta$. To illustrate that, in Fig. \ref{fig:beta}, we 
plot the task-based flow time + energy with and without MapReduce constraints with increasing sizes of map tasks while holding the size of the reduce tasks constant. In particular, we let the size of the reduce tasks to be exponentially distributed with mean $1$, and the size of the map tasks to be exponentially distributed with variable mean $\mu_m$. 
For Fig. \ref{fig:beta}, the number of jobs arriving is Poisson distributed with mean $.5$, and for each job, the number of map and reduce slots are Poisson distributed  with mean $2$ and $5$, respectively.
We see that unlike the theoretical results, the competitive ratio between the algorithm with and without enforcing the MapReduce constraints, does not increase with the ratio of the largest map/reduce task and the smallest reduce task.

\begin{figure}
\centering
\begin{tikzpicture}
    \begin{axis}[
        width  = 0.45*\textwidth,
        height = 8cm,
        major x tick style = transparent,
        ybar,
        bar width=14pt,
        ymajorgrids = true,
        ylabel = {Job Flow Time + Energy ($\times 10^5$)},
        symbolic x coords={$\lambda=.5$, $\lambda=1$, $\lambda=1.5$, $\lambda=2$, $\lambda=2.5$},
        xtick = data,
        scaled y ticks = false,
        legend cell align=left,
        legend style={
                at={(1,1.05)},
                anchor=south east,
                column sep=1ex}
    ]
        \addplot[style={bblue,fill=bblue,mark=none}]
            coordinates {($\lambda=.5$, .1728) ($\lambda=1$,.5725) ($\lambda=1.5$,1.294) ($\lambda=2$,2.1422) ($\lambda=2.5$,3.5429)};

        \addplot[style={rred,fill=rred,mark=none}]
            coordinates {($\lambda=.5$, .0532) ($\lambda=1$,.1639) ($\lambda=1.5$,.3308) ($\lambda=2$,.6397) ($\lambda=2.5$,1.0593) };

%        \addplot[style={ggreen,fill=ggreen,mark=none}]
%            coordinates {($\lambda=.5$,0.92) (Hover,0.56) };
%
%        \addplot[style={ppurple,fill=ppurple,mark=none}]
%            coordinates {($\lambda=.5$,0.74) (Hover,1.07)};

        \legend{With Precedence Const., Without Precedence Const.}
    \end{axis}
\end{tikzpicture}
\caption{Comparison of job-based flow time + energy with and without precedence constraints.}
\label{fig:job}
\end{figure}
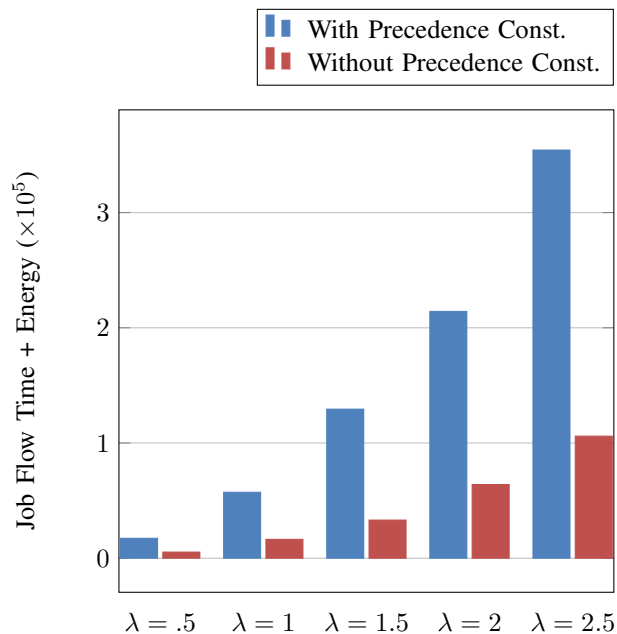

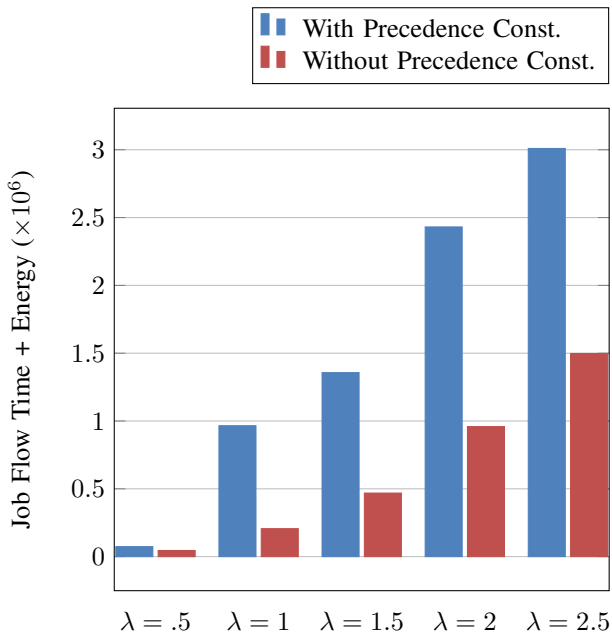
\begin{figure}
\centering
\begin{tikzpicture}
    \begin{axis}[
        width  = 0.45*\textwidth,
        height = 8cm,
        major x tick style = transparent,
        ybar,
        bar width=14pt,
        ymajorgrids = true,
        ylabel = {Job Flow Time + Energy ($\times 10^6$)},
        symbolic x coords={$\lambda=.5$, $\lambda=1$, $\lambda=1.5$, $\lambda=2$, $\lambda=2.5$},
        xtick = data,
        scaled y ticks = false,
        legend cell align=left,
        legend style={
                at={(1,1.05)},
                anchor=south east,
                column sep=1ex}
    ]
        \addplot[style={bblue,fill=bblue,mark=none}]
            coordinates {($\lambda=.5$, .07359) ($\lambda=1$,.9652) ($\lambda=1.5$,1.3569) ($\lambda=2$,2.4310) ($\lambda=2.5$,3.0096)};

        \addplot[style={rred,fill=rred,mark=none}]
            coordinates {($\lambda=.5$, .0453) ($\lambda=1$,.2062) ($\lambda=1.5$,.4686) ($\lambda=2$,.9588) ($\lambda=2.5$,1.4979) };

%        \addplot[style={ggreen,fill=ggreen,mark=none}]
%            coordinates {($\lambda=.5$,0.92) (Hover,0.56) };
%
%        \addplot[style={ppurple,fill=ppurple,mark=none}]
%            coordinates {($\lambda=.5$,0.74) (Hover,1.07)};

        \legend{With Precedence Const., Without Precedence Const.}
    \end{axis}
\end{tikzpicture}
\caption{Comparison of task-based flow time + energy with and without precedence constraints.}
\label{fig:task}
\end{figure}

\begin{figure}
\centering
\begin{tikzpicture}
    \begin{axis}[
        width  = 0.45*\textwidth,
        height = 8cm,
        major x tick style = transparent,
        ybar,
        bar width=14pt,
        ymajorgrids = true,
        ylabel = {Job Flow Time + Energy ($\times 10^5$)},
        symbolic x coords={$\mu_m=2$, $\mu_m=4$, $\mu_m=6$, $\mu_m=8$},
        xtick = data,
        scaled y ticks = false,
        legend cell align=left,
        legend style={
                at={(1,1.05)},
                anchor=south east,
                column sep=1ex}
    ]
        \addplot[style={bblue,fill=bblue,mark=none}]
            coordinates {($\mu_m=2$, .05135) ($\mu_m=4$,.8918) ($\mu_m=6$,1.0697) ($\mu_m=8$,1.4754) };

        \addplot[style={rred,fill=rred,mark=none}]
            coordinates {($\mu_m=2$, .04985) ($\mu_m=4$,.7204) ($\mu_m=6$,.7992) ($\mu_m=8$,1.0225)  };

%        \addplot[style={ggreen,fill=ggreen,mark=none}]
%            coordinates {($\lambda=.5$,0.92) (Hover,0.56) };
%
%        \addplot[style={ppurple,fill=ppurple,mark=none}]
%            coordinates {($\lambda=.5$,0.74) (Hover,1.07)};

        \legend{With Precedence Const., Without Precedence Const.}
    \end{axis}
\end{tikzpicture}
\caption{Comparison of task-based flow time + energy with and without precedence constraints as a function of size of map task being exponentially distributed with mean $\mu_m$, while the size of map task is exponentially distributed with mean $\mu_r=1$.}
\label{fig:beta}
\end{figure}
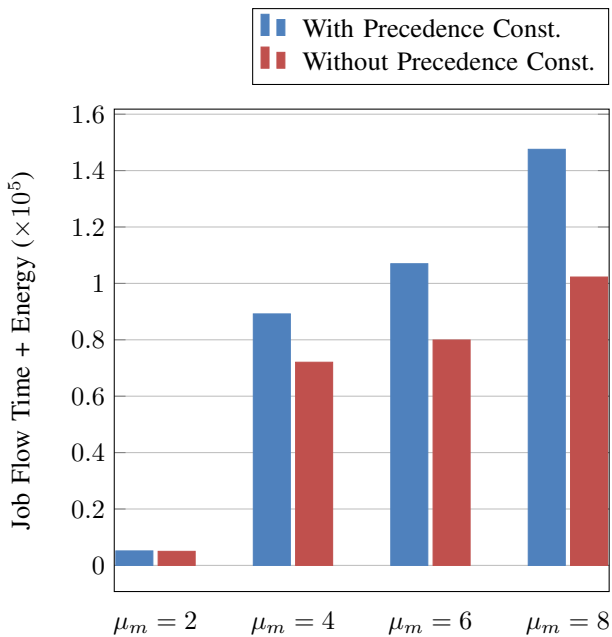

%% file: AppJobFlow.tex
\section{Proof of Lemma \ref{lem:optjobsrpt}}\label{sec:App:JobFlowTime}
An algorithm is called {\it job-work conserving} algorithm, if it processes as many distinct jobs as possible, and multiple tasks of any one job are processed on 
multiple machines only if the number of jobs is less than $K$. 
Thus, if $n(t)$ is the number of outstanding job, where each job has multiple map and reduce tasks, then the algorithm definitely processes some task of $\min\{n(t),K\}$ distinct jobs on the $\min\{n(t),K\}$ servers.

Let $A(j, t)$ denote the amount of work (total volume of jobs processed) completed by a job-work conserving algorithm $A$ for job $j$ by time $t$. 
For set of jobs $J$, $A(J,t) = \sum_{j\in J} A(j, t)$. 

\begin{lemma}\label{lem:scalingsrpt} For a job arrival sequence $\sigma$, at any time $t$, 
let $1\le \rho \le (2-1/K)$, and $\gamma= \frac{2-1/K}{\rho}$. For any job-work conserving algorithm $A$ with $K$ servers having speed $\gamma$, $A(\sigma,\rho t) \ge A'(\sigma, t)$ for any algorithm $A'$ (including $\opt$) using $K$ speed $1$ servers.
\end{lemma}
 Thus, the amount of work completed by any job-work conserving algorithm by time $(2-1/K)t$ is at least as much as work done by $\opt$ by time $t$.  
 The proof is identical to \cite{phillips1997optimal}, that proves an identical result when each job has a single task. The key idea that makes this work is that the algorithm is job-work conserving, to ensure that the algorithm processes as many distinct jobs as possible at any given time. Clearly, it does not hold if an algorithm processes more than one task of a job on multiple servers simultaneously and provides no processing for any task of an outstanding job.

With all jobs being available at time $0$, the input sequence $\sigma$ is essentially the set of jobs (that consists of map and reduce tasks). For the job-SRPT algorithm the following property is easy to prove. 
\begin{proposition}\label{lem:monotone} If the number of jobs finished completely by job-SRPT algorithm is $k$ by time $t$ with input sequence $\sigma'$. Then the number of jobs finished completely by job-SRPT algorithm is at least $k$ by time $t$ with input sequence $\sigma'$, where $\sigma' \subseteq \sigma''$.
\end{proposition}

\begin{proof}[Proof of Lemma \ref{lem:optjobsrpt}]
First note that the job-SRPT algorithm is a job-work conserving algorithm.
Let $\sigma_t\subseteq\sigma$ be the set of jobs that the $\opt$ completely finishes by time $t$ when the input sequence is $\sigma$. Therefore from Lemma \ref{lem:scalingsrpt},  
job-SRPT algorithm finishes all $|\sigma_t|$ jobs by time $t (2-1/K)/\gamma$ given speed $\gamma$ for each server if the input sequence is $\sigma_t$. Now using Lemma \ref{lem:monotone}, we let $\sigma_t = \sigma'$ and $\sigma = \sigma''$ to conclude that for any time $t$, at least $|\sigma_t|$ jobs will be completed by time $(2-1/K) t/\gamma $ with the job-SRPT algorithm with input sequence $\sigma$. 
This implies that for any $k$, the completion time of the $k^{th}$ job with  job-SRPT algorithm is no later than the $(2-1/K)/\gamma$ times the completion time of the  $k^{th}$ job with the $\opt$ for any job arrival sequence $\sigma$. Recall that the flow time of job $j$ is $F_j = T_j-a_j$, where $T_j$ is the departure time of the last reduce 
task of job $j$ and $a_j$ is the arrival time of the job $j$.
Note that the order of departure of jobs with the job-SRPT algorithm and $\opt$ might be different, but since 
flow time is $\sum_{j \in \sigma} (T_j - a_j)$ it is sufficient to show that  
$\sum_{j \in \sigma} F_j(\opt) =  \sum_{j \in \sigma} F_j(\text{job-SRPT})$ to claim that $F_{\text{job}}(\opt) = F_{\text{job}}(\text{job-SRPT})$. 
Choosing $\gamma=2-1/K$, we get that $\sum_{j \in \sigma} F_j(\opt) =  \sum_{j \in \sigma} F_j(\text{job-SRPT})$ and the claim follows.
\end{proof}

%% file: applemjump1.tex
\section{Proof of Lemma \ref{lem:jump1}}
\label{app:lemjump1}
\begin{proof}
Let the total size of a job $j$ be $w_j$ which is the sum of the sizes of all its map and reduce tasks. Recall that map/reduce tasks of at most the last $K-1$ jobs among the total $\sfJ$ are deleted with the $\opt$ that is executing job-SRPT algorithm. Thus, on any map/reduce tasks deletion of job $j$, the increase in $d(t,q) \le 1 $ for $q\le w_j$. Thus, the increase in $\Phi_1(t)$ is at most 
$c_1 \int_0^{w_j} f (d(t,q)+1) dq - \int_0^{w_j} f (d(t,q)) dq \le c_1 \int_0^{w_j} \Delta(d(t,q)) dq\le c_1  \Delta(\frac{\sfJ}{K}) \int_0^{w_j} dq \le c_1 w_j \Delta(\frac{\sfJ}{K})$. 
Since this happens for at most $K-1$ jobs, the total increase in $\Phi_1(t)$ at points of discontinuities is $\sum_{j=1}^{K-1} c_1 w_j \Delta(\frac{\sfJ}{K})$.
The increase in $\Phi_2(t)$ because of deletion of map and reduce task of the $j^{th}$ among the last $K-1$ jobs is simply $c_2 \int_0^{w_j} 1 dq \le c_2 w_j$.  
Since this happens for at most $K-1$ jobs, the total increase in $\Phi_2(t)$ at points of discontinuities is $\sum_{j=1}^{K-1}c_2 w_j $. 
\end{proof}

%% file: Lemmajumpphi3.tex
\section{Proof of Lemma \ref{lem:jump3}}
\begin{proof}
On an arrival of a new job, the new addition to the two terms of \eqref{defn:new:phi3} corresponding to the new map tasks is identical, thus keeping 
$\Phi_3(t)$ unchanged. 

Whenever a map task $j$ is completed by the algorithm or $\opt$, 
  $n(t,q)$ or $n_j^o(t,q)$ is changed for only a single point of $q=0$, keeping the integral unchanged, and the term corresponding to map task $j$ vanishes continuously in \eqref{defn:new:phi3}, and there is no discontinuity. 
  Map task $j$'s departure with the algorithm ($\opt$), however, reduces ${\hat b}_k(t)$ 
  (${\hat b}^o_k(t)$) by $1$ for map task $k$ that belong to the same job as map task $j$ if $b_k(t)\le K$ 
  ($b^o_k(t)\le K$). 
  If a map task departs with the $\opt$, the change in $\Phi_3$ is negative, and we disregard this change, and only consider the increase because of a map task departure with the algorithm.
  %When a map task $j$ departs with the algorithm, the denominator of the terms corresponding to map task $k$ that belong to the same job as map task $j$ decreases by $1$ if their $b_k(t) \le K$, otherwise their is no change.
For a job $i$, departure of its at most $K$ map tasks with the algorithm result in discontinuities (since $b_k(t)\le K$ for any discontinuity to arise), where for the $k^{th}$ discontinuity, there are at most $k$ map tasks of job $i$ remaining in the system. 
Thus, the increase in $\Phi_3(t)$ for any such discontinuity is bounded as follows. Thus, only when $b_k(t)\le K$,
$$\Phi_3(t^+) - \Phi_3(t) = \sum_{k=1}^{b_k(t)} c_3 \delta_k,$$ where 
 \begin{align*}
 \delta_k&=  
  \int g_{z_k({\hat b}, t)} \left( z_k({\hat b}, t) h_k(t, q)\right) dq, \\
  & \quad - c_3  \int g_{z_k({\hat b}+1, t)} \left( z_k({\hat b}+1, t) h_k(t, q)\right) dq, \\
 %   & \stackrel{(a)}\leq \int 
   %g_{b_k(t)+1}\left(\frac{\ell_j +b_k(t)+n_j(t, q)+1}{b_k(t)+1}\right) dq \\ 
   %&  \quad \quad \quad -\int
     % g_{b_k(t)+1}\left(\frac{\ell_j +n_j(t, q)}{b_k(t)+1}\right) dq \\
    %&= \int \Delta\left(\frac{\ell_j +n_j(t, q) +b_k(t)+1}{b_k(t)+1} \right) dq \\
    &\stackrel{(a)}\leq  \Delta\left(\frac{\ell_j +b_k(t)}{b_k(t)}\right)\int_0^{w_{\max}^t} 1 dq, \\
    & = \Delta\left(\frac{\ell_j +b_k(t)}{b_k(t)}\right)w_{\max}^t,
  \end{align*}
  where $(a)$ follows by dropping the second term, and using the definition of $f$ and fact that  $h_k(t, q)\le 1$ for all $q\le w_{\max}^t$, where $w_{\max}^t$ is an upper bound on the size of any map task.
  %, and  ($a$) follows from first part of Lemma~\ref{lemma:f_prop2}, and $(b)$ follows since $n_j(t, q)\le 1$ for all $q$.
  
Thus, for any one map task departure for a fixed value of $b_k(t)$, $$ \Phi_3(t^+) - \Phi_3(t) = c_3w_{\max}^t b_k(t) \Delta \left(\frac{\ell_j +b_k(t)}{b_k(t)}\right),$$ where $\Delta(x) = o(x)$ by Assumption \ref{ass:Pconvex}. In particular,  $\Delta(x) = \alpha x^{1-1/\alpha}$ for $P(s) = s^\alpha$. Since at most $K$ such map task departures corresponding to $b_k(t)=1,\dots, K$ give rise to discontinuties, counting for all $K$ such events we have that for any one job, the maximum increase in $\Phi_3(t)$ is  
$c_3w_{\max}^t \sum_{b_k(t)=1}^K b_k(t) \Delta \left(\frac{\ell_j +b_k(t)}{b_k(t)}\right) = O(c_3 w_{\max}^t K^2  \Delta \left(\ell_j +1\right))$. Thus, counting for all jobs, we get that 
total increase in $\Phi_3(t)$ is $O(\sfJ c_3 w_{\max}^t K^2  \Delta \left(\ell_j +1\right))$, where $\sfJ$ is the total number of jobs. Following Assumption \ref{ass:size}, we get that 
total increase in $\Phi_3(t)$ is $ o(\sum_{j \in \cJ} (\sfm_{jk} + \sfr_{j\ell}))$, as $w_{\max}^t, K$ is fixed, and $\sum_{j=1}^{\sfJ}\ell_j = \sfn_r$.
\end{proof} 

%% file: appfinalcasetaskbased.tex
\section{Proof for remaining case of Theorem \ref{thm:online}}\label{app:finalcasetaskbased}
When the algorithm has $n_f<K$ and $n_m =0$, i.e., there less than $K$ free tasks and all of them are of reduce type, i.e., $n = n_{fr}$,   we only count the positive terms (because of $\opt$) of $d\Phi_3/dt$ of Lemma \ref{lem:drfitspecial} and ignore the negative contribution because of the execution of the algorithm.
Combining Lemma \ref{lemma:phi1_OPT-SRPT}, Lemma \ref{lemma:phi2_OPT-SRPT}, and Lemma \ref{lem:drfitspecial}, 
\begin{align*}
& n + \sum_{k=1}^KP(s_k) + d\Phi_1/dt + d\Phi_2/dt + d\Phi_3/dt   \\ 
&\quad \le  n_{fr}+ n_{fr} P(P^{-1}(1)) +c_1 n^o - c_1 \frac{n(n+1)}{2K}, \\ 
& \quad + c_1 \sum_{k \in \opt} P({\tilde s}_k)  -c_2 n_{fr} P^{-1}(1) +  c_2\sum_{k \in \opt} P({\tilde s}_k) \\
& \quad  + c_3 (\alpha-1)n^o+ c_3 \sum_{k \in \opt} P({\tilde s}_k), \\
& \le c \left(n^o + \sum_{k \in \opt} P({\tilde s}_k)\right),
\end{align*}
for $c = c_1+c_2+(\alpha-1)c_3$, where $c_2\ge 2$.
Thus, $c$ for which \eqref{eq:mothereq} holds for all the three cases is $c = 2\beta +2+2(\alpha-1)$, completing the proof of Theorem \ref{thm:online}.

%% file: AppProofTaskBasedSS.tex
\section{Proofs of Lemma \ref{lem:driftphi1}, \ref{lem:driftphi2}, and \ref{lem:driftphi3}}
 \begin{proof} [Proof of Lemma \ref{lem:driftphi1}]
Since there are only map tasks in $\Phi_1(t)$, this proof is identical to that of Theorem 1.1 \cite{bansal2009speedconf} for single server speed scaling with no precedence constraints, since map tasks are always free.
  \end{proof}

 \begin{proof} [Proof of Lemma \ref{lem:driftphi2}]
 We consider the case when $J(t)\ge J^o(t)$ since otherwise we can show that $d\Phi_2/dt\le 0$ similar to \cite{bansal2009speedconf}. 
 With $J(t)\ge J^o(t)$, for the algorithm, let $q_c$ be the sum of the sizes of all the map tasks for the job with the  smallest cumulative map size. The $\opt$ might be processing a map or a reduce task. If it is a reduce task, then it does not affect $d\Phi_2/dt$. Otherwise, let $\opt$ process a map task $k$ of job $j$, where the sum of the sizes of the map tasks of job $j$ is $q_c^o$. Note that job $j$ may not be the job with the smallest sum of the map tasks, i.e., $J^o(q_c^o) \ne J^o(t)$. Now we need to separate, the three cases 
 $q_c < q_c^o, q_c > q_c^o, q_c = q_c^o$. 
 
 {\bf Case I: $q_c < q_c^o$}.
 The algorithm is processing the shortest (cumulative map task size) job,  thus making $J(q) = J(q)-1$ for $q \in [q_c - s_J(t) dt, q_c]$. Thus, $d\Phi_2(t)/dt$ because of processing by the algorithm is $f(\sfr J(q_c) - 
 \sfr J^o(q_c)-1) - f(\sfr J(q_c) - \sfr J^o(q_c)) s_J = -\Delta(\sfr J(q_c) - \sfr J^o(q_c))s_J$. Now note that $J(q_c) = J(t)$ since the algorithm is processing the shortest (cumulative map task size) job, while $J^o(q_c)\le J^o(t)$ always. Hence 
 $$d\Phi_1(t)/dt \le \Delta(\sfr J(q_c) - \sfr J^o(q_c))s_J  = -\Delta(\sfr J(t) - \sfr J^o(t))s_J.$$
 
Processing by $\opt$ makes $J^o(q) = J^o(q)-1$ for $q \in [q_c^o - {\tilde s}(t) dt, q_c^o]$.
Hence, the $\opt$'s contribution, 
\begin{align*} 
d\Phi_1(t)/dt & \le f(\sfr J(q_c^o) - \sfr J^o(q_c^o)+1) - f(\sfr J(q_c^o) - \sfr J^o(q_c^o)) {\tilde s}, \\
&  = \Delta(\sfr J(q_c^o) - \sfr J^o(q_c)+1){\tilde s}. 
\end{align*} Since $\opt$ is not necessarily processing the map task for a job
 that has the smallest cumulative map task size of all the outstanding map jobs, $J^o(q_c^o) \ne J^o(t)$ and the best bound we can get is  $J^o(q_c^o) \ge 1$. Moreover, since, $J(q_c^o) \le J(t)$ always, we get the $\opt$'s contribution as
 $d\Phi_2(t)/dt \le \Delta(\sfr J(t) - 1)(-s_J+ {\tilde s})$. 
 
 Hence, the argument inside $\Delta$ function for the algorithms' and the $\opt$'s contribution are not identical and we have to apply Lemma \ref{lem:bansal} separately on the algorithms' and the $\opt$'s contribution for $d\Phi_2/dt$ to get $ d\Phi_2/dt \le c_2 P({\tilde s}(t)) - c_2(\alpha-2) \sfr J(t)  + c_2(2-\alpha) \sfr J^o(t)$ by using the fact that the speed of the algorithm is $s_J(t) = P^{-1}(\sfr J(t)+1)$.

The case of $q_c > q_c^o$ and  $q_c = q_c^o$ follows similarly. 

\end{proof}
 
 \begin{proof}[Proof of Lemma \ref{lem:driftphi3}]
 We consider the case when $r(t)\ge r^o(t)$ since otherwise we can show that $d\Phi_3/dt\le 0$ similar to \cite{bansal2009speedconf}. 
 With $r(t)\ge r^o(t)$, let $q_r$ ($q_r^o$) 
 be the size of the smallest free reduce task remaining with the algorithm ($\opt$). Note that with $\opt$, all reduce jobs are always free. 
 Now we need to separate, the three cases 
 $q_r < q_r^o, q_r > q_r^o, q_r = q_r^o$. 
 
 {\bf Case I: $q_r < q_r^o$}.
 The algorithm is processing the shortest free reduce task, thus making $r(q) = r(q)-1$ for $q \in [q_r - s_r(t) dt, q_r]$ while the processing by $\opt$ makes $r^o(q) = r^o(q)-1$ for $q \in [q_r^o - {\tilde s}(t) dt, q_r^o]$ if the $\opt$ is processing a reduce task. If $\opt$ is processing a map task, then it does not affect  $d\Phi_3(t)/dt$. Thus, $d\Phi_3(t)/dt$ because of processing by the algorithm is $\left(f(r(q_r) - r^o(q_r)-1) - f(r(q_r) - r^o(q_r))\right) s_r = -\Delta(r(q_r) - r^o(q_r))s_r$. Now note that $r(q_r) \ge r_f(t)$ since the algorithm is processing the shortest free reduce task, while $r^o(q_r)\le r^o(t)$ always. Hence 
 $$d\Phi_3(t)/dt \le \Delta(r(q_r) - r^o(q_r))s_r  = -\Delta(r_f(t) - r^o(t))s_r.$$
 
 Similarly for the $\opt$'s contribution, 
 \begin{align*} d\Phi_3(t)/dt & \le \left(f(r(q_r^o) - r^o(q_r^o)+1) - f(r(q_r^o) - r^o(q_r^o))\right) {\tilde s}, \\
 &  = \Delta(r(q_r^o) - r^o(q_r)+1){\tilde s},\\  
 &\le \Delta(r_f(t) - r^o(t)){\tilde s},
 \end{align*} as $r(q_r^o)\le r(t)-1$ since $q_r < q_r^o$. Therefore, combining the algorithm's and $\opt$'s contribution, we get 
 $d\Phi_3(t)/dt \le \Delta(r_f(t) - r^o(t))(-s_r+ {\tilde s})$. Using Lemma \ref{lem:bansal}, and the fact that 
 $s_r = P^{-1}(r_f(t)+1)$, we get $d\Phi_3(t)/dt \le c_3 P({\tilde s}(t)) - c_3(r_f(t) - r^o(t))$ as required. 
 
{\bf Case II $q_r > q_r^o$:} The proof for this case is more non-trivial, and essentially reflects why the competitive ratio guarantee holds only for $\alpha <2$. 
In this case, the algorithm's contribution is 
 $$d\Phi_3(t)/dt \le \Delta(r(q_r) - r^o(q_r))s_r  = -\Delta(r_f(t) - r^o(t)+1)s_r,$$
 since $r^o(q_r) \le r^o(t)-1$ as $q_r > q_r^o$.
The $\opt$'s contribution 
\begin{align*} d\Phi_3(t)/dt & \le \left(f(r(q_r^o) - r^o(q_r^o)+1) - f(r(q_r^o) - r^o(q_r^o))\right) {\tilde s}, \\
& = \Delta(r(q_r^o) - r^o(t)+1){\tilde s} .\end{align*} Ideally we would want $r(q_r^o) = r_f(t)$, however, that need not be true, since $q_r > q_r^o$ and there can be caged reduce tasks that are smaller in size than the free reduce jobs in which case $r(q_r^o) = r(t)$, where $r(t)$ is the total number of reduce jobs with the algorithm. Thus, the $\opt$'s contribution is 
$d\Phi_3(t)/dt \le \Delta(r(t) - r^o(t)){\tilde s}$. Hence, the argument inside $\Delta$ function for the algorithms' and the $\opt$'s contribution are not identical and we have to apply Lemma \ref{lem:bansal} separately on the algorithms' and the $\opt$'s contribution for $d\Phi_3/dt$ to get $ d\Phi_3/dt \le c_3 P({\tilde s}(t)) - c_3(r_f(t) - (\alpha-1)r(t) - 
 (2-\alpha)r^o(t))$, where we have used the fact that the speed of the algorithm for processing the shortest free reduce task is $P^{-1}(r_f(t)+1)$.

The case for $q_r = q_r^o$ follows similarly. 
 \end{proof}

%% file: lemlpnorm.tex
\section{Lemma \ref{lem:lpnorm}}
\begin{lemma}\label{lem:lpnorm}
For $\alpha \ge 1$, and $x_i\ge 0$, $\left(\sum_{i=1}^n x_i^{1/\alpha}\right)^\alpha \le n^{\alpha-1} \sum_{i=1}^n x_i.$
\end{lemma}
\begin{proof}
Let $X$ be a random variable with distribution $D$ over the support $x_1, \dots, x_n$. Let $\eta > \nu > 0$. 
Then
$\bbE\{|X|^\eta\}^{1/\eta} = \bbE\{(|X|^\nu)^{\eta/\nu}\}^{1/\eta} \ge \bbE\{(|X|^\nu)\}^{\frac{\eta}{\nu} \frac{1}{\eta}} = \bbE\{(|X|^\nu)\}^{1/\nu}$
where the inequality follows from Jensen's inequality since $x^{\eta/\nu}$ is convex. Choosing $D$ to be the uniform distribution, and $\eta =1$ and $\nu = 1/\alpha$ with $\alpha\ge 1$, we get the inequality.
\end{proof} 

%% file: AppMSalpha.tex
\section{Proof of Theorem \ref{thm:onlinems}}\label{app:onlinems}
The potential function we construct as follows is essentially the multi-server version of \eqref{defn:phiss} as proposed in \cite{vaze} for analyzing the multi-server SRPT algorithm, except for the 
$\Phi_4$ term that is needed to handle the unique challenge posed by the precedence constraints described in the previous paragraph. 
Let
$d_m(q) = \max\left\{0,\frac{m(q) - m^o(q)}{K}\right\}$, $d_J(q) = \max\left\{0,\sfr \frac{(J(q) - J^o(q))}{K}\right\}$, $d_r(q) = \max\left\{0,\frac{r(q) - r^o(q)}{K}\right\}$.
We consider the potential function
\begin{equation}\label{defn:phims}
\Phi(t) = \Phi_1(t)+\Phi_2(t)+\Phi_3(t)+\Phi_4(t),
\end{equation}
where $\Phi_i(\cdot)$ are defined as follows.
\begin{align*}\label{defn:phims1} 
  \Phi_1(t) &= c_{11} \int_{0}^\infty f\left(d_m(q)\right) dq + 
 c_{12} \int_{0}^\infty (m(q) - m^o(q)) dq,  \\
 \Phi_2(t) &= c_{21}\int_{0}^\infty f\left(d_J(q)\right) dq
  \end{align*}
\begin{equation*}\label{defn:phims3} 
  \Phi_3(t) = c_{31}\int_{0}^\infty f\left(d_r(q)\right) dq + c_{32} \int_{0}^\infty (r(q) - r^o(q)) dq.
\end{equation*}
The sub-potential function $\Phi_4(t)$ is special and needed to handle the case when the number of jobs with the algorithm is less than $K$.
For the $j^{th}, j=1, \dots, J(t) (J^o(t))$ job, let $q_j$ ($q_j^o$) be the sum of the remaining size of all its map tasks with the algorithm and the $\opt$, respectively. Then for the $j^{th}$ job $n_j(t,q)=1$ ($n_j^o(t,q)=1$) for $q \le q_j$ ($q \le q_j^o$) 
and zero otherwise.

Let $ \Phi_4(t) $
\begin{equation*}\label{defn:phims4} 
 = c_4 \left(\sum_{j \in J(t)} \int_{0}^\infty f\left(\sfr n_j(q)\right) dq -  \sum_{j \in J^o(t)} \int_{0}^\infty f\left(\sfr n_j^o(q)\right) dq\right).
\end{equation*}
$c_{11},c_{12}, c_{21}, c_{22}, c_{31}, c_{32},c_4$ are positive constants to be determined later.

The drift of $\Phi_1, \Phi_2, \Phi_3$ can be computed similar to Lemma 10 and 11 \cite{vaze}, when $\opt$ has no precedence constraints but is not restricted to perform multi-server SRPT scheduling. 
\begin{lemma}\label{lem:msphi1} $d\Phi_1/dt $
  For $m(t) \geq K$, \begin{align*} 
  \le &c_{11} m^o -
  c_{11}(2-\alpha) m + c_{11} (2-\alpha)\left(\frac{K-1}{2}\right)\\ &\qquad
  + c_{11}\sum_{k \in \opt} P({\tilde s}_k) -c_{12} \min(K,m) + c_{12} \sum_{k \in \opt} P({\tilde s}_k) \end{align*} while for $m(t) <
  K,$ \begin{align*} d\Phi_1/dt &\le c_{11} m^o
  + \frac{c_{11}(2-\alpha)m}{2} + c_{11}\sum_{k \in \opt} P({\tilde
  s}_k), \\
  &\quad  -c_{12} \min(K,m) + c_{12}\sum_{k \in \opt} P({\tilde s}_k) \end{align*}
\end{lemma}

\begin{lemma}\label{lem:msphi2}
  For $J(t) \geq K$, \begin{align*} 
  d\Phi_2/dt \le &c_{21} \sfr J^o -
  c_{21}(2-\alpha) \sfr J + c_{21} (2-\alpha)\left(\frac{K-1}{2}\right) \\ 
  & \quad + c_{21}\sum_{k \in \opt} P({\tilde s}_k)  \end{align*} 
  while for  $J(t) < K$, (by disregarding the algorithm's contribution)
  \begin{align*} 
  d\Phi_2/dt \le &c_{21} \sfr J^o
   + c_{21} (2-\alpha)\left(\frac{K-1}{2}\right)+ c_{21}\sum_{k \in \opt} P({\tilde s}_k)  \end{align*}
\end{lemma}
\begin{lemma}\label{lem:msphi3}
For $r_f(t) \geq K$,  $  d\Phi_3/dt$ \begin{align*} 
 \le &c_{31} r^o -
  c_{31}(2-\alpha) r_f + c_{31} (2-\alpha)\left(\frac{K-1}{2}\right)\\ &
  + c_{31}\sum_{k \in \opt} P({\tilde s}_k) -c_{32} \min(K,m) + c_{32} \sum_{k \in \opt} P({\tilde s}_k) \end{align*} while for 
  $r_f(t) <K,$ \begin{align*} d\Phi_3/dt &\le c_{31} r^o
  + \frac{c_{31}(2-\alpha)r_f}{2} + c_{31}\sum_{k \in \opt} P({\tilde
  s}_k), \\
  & \quad -c_{32} \min(K,m) + c_{32}\sum_{k \in \opt} P({\tilde s}_k) \end{align*}
\end{lemma}

The only non-trivial drift is for $\Phi_4(t)$ which is needed when $J(t) < K$  is described as follows.

\begin{lemma}\label{lem:msphi4} When $$d\Phi_4(t)/dt \le c_4 ( - \b1_{J(t) < K}  \ \sfr J(t) + \sfr \min\{J^o(t), K)). $$
\end{lemma}
Proof is similar to Lemma \ref{lem:drfitspecial} and hence omitted. 

\begin{proof}[Proof of Theorem \ref{thm:onlinems}]
To prove the Theorem we will show that \eqref{eq:mothereq} holds for an appropriate choice of $c$ using Lemma \ref{lem:msphi1}, \ref{lem:msphi2}, \ref{lem:msphi3} and \ref{lem:msphi4}.

For each server, the energy cost is similar to the single server case and is given by
$P(P^{-1}(\frac{m(t)+1}{K})+P^{-1}(\frac{\sfr J(t)+1}{K}) + P^{-1}(\frac{r_f(t)+1}{K})) \le \frac{3^{\alpha-1}}{K} (m(t) + \sfr J(t) + r_f(t)+ 3)$ using Lemma \ref{lem:lpnorm} since $P^{-1}(x) = x^{1/\alpha}$ with $n=3$ and $\alpha \ge 1 $. Thus, summing across all the $K$ servers, the total energy cost $\sum_{k=1}^KP(s_k(t))$ at any time $\le 3^{\alpha-1} (m(t) + \sfr J(t) + r_f(t)+ 3)$. 

Thus, the running cost $n(t) + \sum_{k=1}^KP(s_k(t))$ for the algorithm is  at most $3^{\alpha-1}(m(t) + \sfr J(t) + r_f(t)) + 3^{\alpha}$.
For the case when  $J(t) < K$, $m(t)\ge K$, $r_f\ge K$, combining Lemma \ref{lem:msphi1}, \ref{lem:msphi2}, \ref{lem:msphi3}, \ref{lem:msphi4} we can write \eqref{eq:mothereq}, $  n(t) + \sum_{k=1}^KP(s_k(t)) + d\Phi/dt$
 \begin{align*}
  &  \le  (3^{\alpha-1}+1)(m(t) + \sfr J(t) + r_f(t)) + 3^{\alpha}+ c_{11} m^o\\
  &\quad -
  c_{11}(2-\alpha) m + c_{11} (2-\alpha)\left(\frac{K-1}{2}\right)\\ 
  &\qquad
  + c_{11}\sum_{k \in \opt} P({\tilde s}_k) -c_{12} \min(K,m) + c_{12} \sum_{k \in \opt} P({\tilde s}_k), \\
 & \qquad+ c_{31} r^o -
  c_{31}(2-\alpha) r_f + c_{31} (2-\alpha)\left(\frac{K-1}{2}\right)\\ 
  & \qquad+ c_{31}\sum_{k \in \opt} P({\tilde s}_k) -c_{32} \min(K,m) + c_{32} \sum_{k \in \opt} P({\tilde s}_k),\\ 
  & \qquad + c_4 ( - \b1_{J(t) < K}  \ \sfr J(t) + \sfr \min\{J^o(t), K)),\\
 &  \le c (n^o(t) +  \sum_{k \in \opt} P({\tilde s}_k)),
 \end{align*}
 for $c= c_1+c_2+c_3+3^{\alpha}$, where $c_1 = 4, c_2 = \frac{4}{2-\alpha}, c_3 =4$, since $\alpha<2$
Hence, the competitive ratio of the proposed online algorithm is  $8+\frac{4}{2-\alpha}+3^{\alpha}$.
Other case depending on $J(t)\ge K$ and $r_f<K$ follow similarly.
\end{proof}